\documentclass[letterpaper,11pt]{article}
\usepackage{amsfonts}
\usepackage{amsmath,amsthm}
\usepackage{amssymb,amscd}
\usepackage{enumerate}
\usepackage{cite}

\makeatletter

\newcommand{\Rmnum}[1]{\expandafter\@slowromancap\romannumeral #1@}
\makeatother

{}
\newtheorem{theorem}{Theorem}[section]

\newtheorem{corollary}[theorem]{Corollary}
\newtheorem{lemma}[theorem]{Lemma}
\newtheorem{proposition}[theorem]{Proposition}

\theoremstyle{definition}

\theoremstyle{remark}
\newtheorem{remark}{Remark}[section]

\swapnumbers \makeatletter
\def\ps@mystyles{    \def\@evenhead{\hfil{\sl\leftmark}\hfil}    \def\@oddhead{\hfil{\sl\rightmark}\hfil}    }
\makeatother

\begin{document}
\title{The current distribution of the multiparticle hopping asymmetric diffusion model}

\author{\textbf{Eunghyun Lee}\\ \small{\textit{Department of Mathematics and Statistics}}
                             \\ \small{\textit{University of Helsinki}}
                             \date{}            \\ \small{\textit{Email:eulee@mappi.helsinki.fi}}}
\maketitle

\begin{abstract} \noindent In this paper we treat the \textit{multiparticle hopping asymmetric diffusion model} (MADM) on $\mathbb{Z}$ introduced by Sasamoto and Wadati in 1998. The transition probability of the MADM with $N$ particles is provided by using the Bethe ansatz. The transition probability is expressed as the sum of  $N$-dimensional contour integrals of which contours are circles centered at the origin with restrictions on their radii. By using the transition probability we find $\mathbb{P}(x_m(t) =x)$, the probability that the $m$th particle from the left is at $x$ at time $t$.  The probability $\mathbb{P}(x_m(t) =x)$ is expressed as the sum of $|S|$-dimensional contour integrals over all $S \subset \{1,\cdots,N\}$ with $|S|  \geq m$, and is used to give the current distribution of the system. The mapping between the MADM and the pushing asymmetric simple exclusion process (PushASEP) is discussed.
\end{abstract}
\section{Introduction}\label{intro}
The Bethe ansatz, which is useful to find the eigenvalues and eigenvectors of the Hamiltonian of one-dimensional quantum spin systems, has been one of the main techniques in studying one-dimensional stochastic particle models \cite{Gwa,Lee,Povol,Povol2,Rakos,Sasamoto,Sasamoto2,SasaWada1998Oct,Schutz,TW1}. One direction of studying these models by the Bethe ansatz is  the asymptotic analysis of the current that can be computed from transition probabilities. In this direction the simple exclusion process on $\mathbb{Z}$ has been extensively studied \cite{Lee,Nagao,Rakos,TW1,TW2,TW3,TW4,TW5} and some results are now used to show that  the Kadar-Parisi-Zhang (KPZ)  equation belongs to the KPZ universality class\cite{Amir}.\\ \\
The model we are going to consider in this paper is the \textit{one-dimensional asymmetric diffusion model without exclusion}. This model was also named the \textit{multiparticle hopping asymmetric diffusion model} (MADM) by Sasamoto and Wadati  \cite{SasaWada1998Oct}. The dynamics of the MADM with $N $ particles  is as follows: Each site on $\mathbb{Z}$ is equipped with $2N$ Poisson clocks with rates $pr_n$ and $ql_n~ (p+q=1,~n=1,\cdots,N)$ and all clocks are independent. If a clock with rate $pr_n(ql_n)$  rings at $x$, which is occupied  by at least $n$ particles, then $n$ particles simultaneously jump to $x+1$ ($x-1$). However, if  the number of particle at a site  is less than $n$ when the clock with rate $pr_n(ql_n)$ rings at the site, then nothing happens and the clock resumes.  The rates are not free but are given by the $q$-binomial coefficients
\begin{equation}
 {[n]_{\frac{\lambda}{\mu}}}= \frac{1 - \big(\frac{\lambda}{\mu}\big)^n }{1-\frac{\lambda}{\mu}} = \frac{1}{r_n}\hspace{0.5cm}\label{RightRate}
  \end{equation}
with $q=\frac{\lambda}{\mu}$ and
\begin{equation}
 {[n]_{\frac{\mu}{\lambda}}}= \frac{1 - \big(\frac{\mu}{\lambda}\big)^n }{1-\frac{\mu}{\lambda}} = \frac{1}{l_n} \label{LeftRate1}
  \end{equation}
with $q=\frac{\mu}{\lambda},$
   where $\lambda + \mu=1$ and $\lambda,\mu \neq 0$. These requirements on rates are needed for the Bethe ansatz solvability \cite{Ali,Ali2,SasaWada1998Oct}. In \cite{SasaWada1998Oct}  only one free parameter was involved in the rates  but it is possible to extend to have two free parameters $p$ and $\lambda$ as above \cite{Ali2}. In \cite{Ali2}  the model with two parameters has exclusion property and  the multiparticle hopping in \cite{SasaWada1998Oct} is interpreted as \textit{pushing effect}. In this \textit{pushing} version the model becomes the drop-push model  in the limit $q,\lambda \rightarrow 0$ \cite{Sasamoto2} or more generally the pushing asymmetric simple exclusion process (PushASEP)\footnote{In the PushASEP in \cite{Borodin2}, the \textit{pushing} can occur in only one direction, so we will call it the \textit{one-sided} PushASEP} in the limit $\lambda \rightarrow0$ in \cite{Borodin2}. \\ \\
  The totally asymmetric simple exclusion process (TASEP) and the \textit{one-sided} PushASEP are \textit{determinantal} in the sense that their transition probabilities are expressed as determinants whose entries are contour integrals
  \cite{Schutz, Borodin2}. These transition probabilities are a starting point for the multi-point joint distributions as well as the currents of systems \cite{Borodin2,Nagao,Rakos,Sasamoto}. In the asymmetric simple exclusion process (ASEP) with $N$ particles, the transition probability is not in a determinantal form but it is expressed as a sum of $N$-dimensional contour integrals over all permutations in a symmetric group $\mathbb{S}_N$ \cite{TW1}. Hence, some techniques for determinantal models to study the current of the system                                        are not applicable to the ASEP. In spite of this limitation, Tracy and Widom obtained the current distribution of the ASEP with step initial condition and found the fact that its fluctuation is governed by the GUE Tracy-Widom distribution \cite{TW1,TW2,TW3,TW4,TW5}.  \\ \\
  The main goal of this paper is to take a first step to see  the MADM is also governed by the Tracy-Widom distribution. As the first step, we  provide the transition probability and the current distribution of the MADM with $N$ particles. In fact, the MADM has been expected to belong to the KPZ universality class by studying its energy gap \cite{SasaWada1998Oct} and it is known that the current fluctuation of the \textit{one-sided} PushASEP (equivalently the MADM in the limit $\lambda \rightarrow 0$) with a special initial condition is governed by the Tracy-Widom distribution \cite{Borodin2}. Hence, we expect our results in this paper to provide a starting point to extend a result of the \textit{one-sided} PushASEP to the \textit{two-sided} PushASEP (equivalent to the MADM).  \\ \\
    This paper is organized as follows. In Section \ref{section1} we find the transition probability of the MADM  by using the Bethe ansatz. As in the ASEP \cite{TW1}, the transition probability is given by the sum of $N$-dimensional contour integrals. However, unlike the transition probability of the ASEP in \cite{TW1}, there is a delicate condition on contours  in the MADM.  This difference is due to the difference in the $S$-matrices of the ASEP and the MADM, and the difference also affects the way of obtaining the current distribution. In Section \ref{sectionMth} we provide the probability that the $m$th leftmost particle  is at $x$ at time $t$ in Theorem \ref{main2} and the current distribution of the system in (\ref{maineqcurrent}).    In Section \ref{push} we discuss the relation between the MADM and the \textit{two-sided} PushASEP. The mapping between the MADM and the PushASEP is confirmed by their transition probabilities, and the probability that the $m$th leftmost particle in the \textit{two-sided} PushASEP is at $x$ at time $t$  is provided.
\section{The transition probability of the MADM}\label{section1}
The MADM with $N$ particles is a continuous-time Markov process with a countable state space. If we denote the $m$th leftmost particle's position by $x_m$, then a state of the process can be specified by particles' positions $X = (x_1,\cdots, x_N) \in \mathbb{Z}^N$ with $x_1 \leq \cdots \leq x_N$. We will call the state space
\begin{equation*}
\{(x_1,\cdots, x_N) \in \mathbb{Z}^N : x_1 \leq \cdots \leq x_N\}
\end{equation*}
the \textit{physical region} of the MADM. By the standard technique \cite{Schutz,SasaWada1998Oct,TW1} to use Bethe ansatz we have a differential equation for a function $u(X;t) = u(x_1,\cdots,x_N;t)$ on $\mathbb{Z}^N \times [0, \infty)$
 \begin{eqnarray}
 \frac{d}{dt}~u(X;t) &=& \sum_{i=1}^N\Big[pu(x_1,\cdots,x_{i-1},x_i-1,x_{i+1},\cdots,x_N;t)  \label{master} \\
   & &\hspace{2.5cm}+~ qu(x_1,\cdots,x_{i-1},x_i+1,x_{i+1},\cdots,x_N;t) -~ u(X;t)\Big] \nonumber
\end{eqnarray}
with the condition imposed on the boundary of the physical region
\begin{eqnarray}
 & &u(x_1,\cdots,x_i,x_i-1,\cdots,x_N;t) \label{boundary} \\
 & =&  \mu u(x_1,\cdots,x_i-1,x_i-1,\cdots,x_N;t) + \lambda u(x_1,\cdots,x_i,x_i,\cdots,x_N;t) \nonumber
\end{eqnarray}
for $i=1,\cdots,N$, and the initial condition in the physical region, that is,
\begin{equation}
u(X;0) = \delta_Y(X) \label{IC} ~~~\textrm{for}~~x_1 \leq \cdots \leq x_N,
\end{equation}
where $Y= (y_1,\cdots, y_N)$ is the initial state of the MADM. The solution of (\ref{master}) with (\ref{boundary}) and (\ref{IC}) is the transition probability and we will  denote it by $P_Y(X;t)$. Define the $S$-matrix for the MADM
 \begin{equation*}
S_{\beta\alpha} :=-\frac{\mu + \lambda \xi_{\alpha}{\xi_{\beta}} - \xi_{\alpha}}{\mu + \lambda \xi_{\alpha}{\xi_{\beta}} - \xi_{\beta}}~~(\mu+\lambda =1)
\end{equation*}
and let
\begin{equation}
A_{\sigma} = \prod_{\substack{i<j,\\ \sigma(i) > \sigma(j)}} S_{\sigma(i)\sigma(j)} \label{coeff}
\end{equation}
for the coefficient of the Bethe ansatz solution of (\ref{master}) with (\ref{boundary}). Here, $\sigma$ is a permutation  in a symmetric group $\mathbb{S}_N$ on $\{1,\cdots,N\}$ and the product is over all $S$-matrices that satisfies the given condition  under the product symbol. If there is no $S$-matrix that satisfies the condition, then $A_{\sigma}$ is defined to be 1\footnote{In general, for any product of $S$-matrices with a certain condition, if there is no  $S$-matrix that satisfies the condition, then the product is defined to be 1 in this paper.}. Let
\begin{equation*}
 \varepsilon(\xi_i) = \frac{p}{\xi_i} + q\xi_i -1,
\end{equation*}
which is interpreted as \textit{energy} of the MADM.
\begin{lemma} \label{lemma1}
Let $\lambda + \mu=1 ~(1/2 < \mu < 1)$  and $\alpha < \beta $. Let  $\mathcal{C}_{R_i}$ be a circle centered at 0  with radius $R_i$ in $\mathbb{C}$ and $\xi_i \in \mathcal{C}_{R_i}~(i=\alpha,\beta)$. If ~$1< R_{\alpha} < R_{\beta} < \frac{\mu}{\lambda},$ then
 \begin{equation*}
   R_{\beta}<  \left|\frac{\xi_{\beta} - \mu}{\lambda\xi_{\beta}}\right|.
\end{equation*}
\end{lemma}
\begin{proof}
 It is easy to see that
 \begin{equation*}
 R_{\beta} < \frac{R_{\beta} - \mu}{\lambda R_{\beta}}  \leq \left|\frac{\xi_{\beta} - \mu}{\lambda\xi_{\beta}}\right|
 \end{equation*}
where the first inequality is equivalent to
 \begin{equation*}
  1< R_{\beta} < \frac{\mu}{\lambda}.
\end{equation*}
\end{proof}
\begin{lemma}\label{lemma2}
Let $\lambda + \mu=1~(1/2 < \mu \leq 1)$ and $\alpha < \beta $. Let  $\mathcal{C}_{R_i}$ be a circle centered at 0  with radius $R_i$ in $\mathbb{C}$ and $\xi_i \in \mathcal{C}_{R_i}~(i=\alpha,\beta)$. If $1< R_{\alpha} < R_{\beta} < c$ where
\begin{equation*}
c = \begin{cases}
\frac{\mu}{\lambda} & \textrm{if}~~\lambda \neq 0\\
\infty &\textrm{if}~~ \lambda =0,
\end{cases}
\end{equation*}
then
\begin{equation}
 R_{\alpha} < \left| \frac{ \xi_{\alpha}\xi_{\beta}}{\mu+\lambda\xi_{\alpha}\xi_{\beta}}\right|. \label{lemma2eq}
 \end{equation}
\end{lemma}
\begin{proof}
Let $\lambda \neq 0$. It is easy to see
 \begin{equation*}
    R_{\alpha}< \frac{R_{\alpha}R_{\beta}}{\mu + \lambda R_{\beta}R_{\beta}} < \frac{R_{\alpha}R_{\beta}}{\mu + \lambda R_{\alpha}R_{\beta}}  \leq \left| \frac{ \xi_{\alpha}\xi_{\beta}}{\mu+\lambda\xi_{\alpha}\xi_{\beta}}\right|
 \end{equation*}
where the first inequality is equivalent to
\begin{equation*}
    1 < R_{\beta} < \frac{\mu}{\lambda}.
\end{equation*}
If $\lambda =0$, (\ref{lemma2eq}) is obvious.
\end{proof}
\begin{lemma}\label{lemma3}
Let $\lambda + \mu=1~(1/2 < \mu \leq 1)$ and $\alpha < \beta < \gamma$. Let  $\mathcal{C}_{R_i}$ be a circle centered at 0  with radius $R_i$ in $\mathbb{C}$ and $\xi_i \in \mathcal{C}_{R_i}~(i=\alpha,\beta,\gamma)$.  If~ $1< R_{\alpha} < R_{\beta} < R_{\gamma} < c$ where
\begin{equation*}
c = \begin{cases}
\frac{\mu}{\lambda} & \textrm{if}~~\lambda \neq 0\\
\infty &\textrm{if}~~ \lambda =0,
\end{cases}
\end{equation*}
then
\begin{equation}
  R_{\alpha} <  \left| \frac{\xi_{\alpha}\xi_{\gamma} - \lambda\xi_{\alpha}\xi_{\beta}\xi_{\gamma} }{\mu} \right|.
  \label{lemma3eq}
\end{equation}
\end{lemma}
\begin{proof}
Let $\lambda\neq 0$. It is easy to see that
 \begin{equation*}
  R_{\alpha}< \frac{R_{\alpha}R_{\gamma}(1- \lambda R_{\gamma})}{\mu}<  \frac{R_{\alpha}R_{\gamma}(1- \lambda R_{\beta})}{\mu}\leq   \left| \frac{\xi_{\alpha}\xi_{\gamma} - \lambda\xi_{\alpha}\xi_{\beta}\xi_{\gamma} }{\mu} \right|
 \end{equation*}
where the first inequality is equivalent to
 \begin{equation*}
   1 < R_{\gamma} < \frac{\mu}{\lambda}.
\end{equation*}
If $\lambda =0$, (\ref{lemma3eq}) is obvious.
\end{proof}
\noindent Now, we find the transition probability of the MADM.
\begin{theorem}\label{theorem1}
Let $\lambda + \mu=1~ (1/2 < \mu  \leq 1)$ and $\mathcal{C}_{R_i}~ (i=1,\cdots, N)$ be a circle oriented counterclockwise, centered at 0 with radius $R_i$. Assume that $1< R_{1} <\cdots < R_{N} < c$ where
\begin{equation*}
c = \begin{cases}
\frac{\mu}{\lambda} & \textrm{if}~~\lambda \neq 0\\
\infty &\textrm{if}~~ \lambda =0.
\end{cases}
\end{equation*}
 The transition probability of the MADM with $N \geq 2$ particles is
\begin{equation}
P_Y(X;t) = \sum_{\sigma \in \mathbb{S}_N}\Big(\frac{1}{2\pi i}\Big)^N\int_{\mathcal{C}_{R_N}}\cdots\int_{\mathcal{C}_{R_1}}A_{\sigma}\prod_i^N \Big(\xi_{\sigma(i)}^{x_i - y_{\sigma(i)}-1}e^{\varepsilon(\xi_i)t}\Big)~d\xi_1\cdots \xi_N.\label{transition}
\end{equation}
\end{theorem}
\begin{remark} In this paper we consider the nontrivial MADM with at least two particles because the MADM is just the simple random walk when $N=1$. It can be easily shown that (\ref{transition}) satisfies (\ref{master}) and (\ref{boundary}) for all $X \in \mathbb{Z}^N$ in the same way as in the proof of the ASEP \cite{TW1}. Hence, we shall prove only that (\ref{transition}) satisfies (\ref{IC}).
\end{remark}
\begin{proof}
The initial condition (\ref{IC}) is satisfied if
\begin{equation*}
P_Y(X;0) = \Big(\frac{1}{2\pi i}\Big)^N\int_{\mathcal{C}_{R_N}}\cdots\int_{\mathcal{C}_{R_1}} \prod_i \xi_{i}^{x_i - y_{i}-1}d\xi_1\cdots \xi_N,
\end{equation*}
which is the integral for the identity permutation \textit{id}. Hence we will show that
\begin{equation}
\sum_{\sigma \neq id}\int_{\mathcal{C}_{R_N}}\cdots\int_{\mathcal{C}_{R_1}}A_{\sigma}\prod_i \xi_{\sigma(i)}^{x_i - y_{\sigma(i)}-1}d\xi_1\cdots \xi_N=0 \label{Isigma}
\end{equation}
for any state in the physical region.
We prove (\ref{Isigma}) by induction. First, we  show that
\begin{equation}
 \int_{\mathcal{C}_{R_2}}\int_{\mathcal{C}_{R_1}}\frac{\mu + \lambda\xi_2\xi_1 - \xi_1}{\mu + \lambda\xi_2\xi_1 - \xi_2}~\xi_2^{x_1-y_2-1}\xi_1^{x_2-y_1-1}d\xi_1d\xi_2=0 \label{right2}
\end{equation}
when $x_1 \leq x_2$ and $y_1 \leq y_2$.
 Substituting $\xi_2= \frac{\eta}{\xi_1}$ so that $\eta$ runs over $ \mathcal{C}_{R}$ where $R= R_1R_2$, the left hand side of (\ref{right2}) becomes
\begin{equation}
\int_{\mathcal{C}_{R}}\int_{\mathcal{C}_{R_1}}
\frac{\mu + \lambda\eta - \xi_1}{\xi_1(\mu + \lambda\eta) - {\eta}}~\eta^{x_1-y_2-1}\xi_1^{x_2-x_1+y_2-y_1}d\xi_1d\eta \label{Forpush}
\end{equation}
where  $x_2-x_1 + y_2- y_1 \geq 0$. If we integrate with respect to $\xi_1$, the integrand has a simple pole at $\frac{\eta}{\mu + \lambda\eta}$, which is the only pole. This pole lies outside of $\mathcal{C}_{R_1}$ by Lemma \ref{lemma2} and so the integral  with respect to $\xi_1$ is zero.\\
 \indent Now, suppose that the statement is true for $N= K-1$. Let
   \begin{equation*}
      I(\sigma) =\int_{\mathcal{C}_{R_K}}\cdots\int_{\mathcal{C}_{R_1}}A_{\sigma}\prod_i^K \xi_{\sigma(i)}^{x_i - y_{\sigma(i)}-1}~d\xi_1\cdots \xi_K.
   \end{equation*}
   The sum of $I(\sigma)$ over $\sigma \in \mathbb{S}_K$ such that $\sigma(K) = K$ and $\sigma \neq id$ is simply
   $$ \int_{\mathcal{C}_{R_K}}\Bigg[\sum_{\substack{\sigma \neq id,\\ \sigma(K) =K}} \int_{\mathcal{C}_{R_{K-1}}}\cdots\int_{\mathcal{C}_{R_1}}A_{\sigma}\prod_i^{K-1} \xi_{\sigma(i)}^{x_i - y_{\sigma(i)}-1}~d\xi_1\cdots \xi_{K-1}\Bigg] \xi_K^{x_K -y_K-1}d\xi_K = 0$$
    by the induction hypothesis. Now we  will show that  $I(\sigma) =0$ for each $\sigma \in \mathbb{S}_K$ such that  $ \sigma(k)=K$ where $k \neq K$. Let $Q_{\sigma}= \{\sigma(k+1), \cdots, \sigma(K)\}$ and $q= \min Q_{\sigma}$. For a fixed $k$, observe that
    \begin{equation*}
     A_{\sigma} = \prod_{\substack{i<j,k,\\\sigma(i) > \sigma(j)}} S_{\sigma(i)\sigma(j)}
                  \prod_{k<j } S_{K\sigma(j)}
                  \prod_{\substack{k<i<j,\\\sigma(i) > \sigma(j)}} S_{\sigma(i)\sigma(j)}.
    \end{equation*}
   Substituting $\xi_K = \frac{\eta}{\prod_{i=1}^{K-1} \xi_i}$ so that $\eta \in \mathcal{C}_R$ where $R = \prod_i^KR_i$ in
  \begin{eqnarray*}
  I(\sigma) &=& \int_{\mathcal{C}_{R_K}}\cdots\int_{\mathcal{C}_{R_1}}\prod_{\substack{i<j,k,\\\sigma(i) > \sigma(j)}} S_{\sigma(i)\sigma(j)}
                  \prod_{k<j} S_{K\sigma(j)}
                  \prod_{\substack{k<i<j,\\\sigma(i) > \sigma(j)}} S_{\sigma(i)\sigma(j)} \\
  & & ~~~ \Big(\prod_{i \neq k, \sigma^{-1}(q)} \xi_{\sigma(i)}^{x_i-y_{\sigma(i)} -1} \Big)~\xi_K^{x_k -y_K-1}\xi_q^{x_{\sigma^{-1}(q)} - y_q-1}~d\xi_1 \cdots d\xi_K,
  \end{eqnarray*}
  and ignoring the sign of $S$-matrices, we have
   \begin{eqnarray}
   & & \int_{\mathcal{C}_{R}}\int_{\mathcal{C}_{R_{K-1}}}\cdots \int_{\mathcal{C}_{R_1}}\prod_{\substack{i<j,k,\\\sigma(i) > \sigma(j)}} S_{\sigma(i)\sigma(j)}\prod_{\substack{k<i<j,\\\sigma(i) > \sigma(j)}} S_{\sigma(i)\sigma(j)} \label{eqForPush} \\
  &  &\Bigg(\frac{\mu+\frac{\lambda\eta}{\prod_{i\neq K,q}\xi_i} - \xi_q}{\xi_q\big(\mu+\frac{\lambda\eta}{\prod_{i\neq K,q}\xi_i}\big) - \frac{\eta}{\prod_{i\neq K,q}\xi_i} }\Bigg)
  \prod_{\substack{k<j,\\ j \neq \sigma^{-1}(q)}}  \Bigg(\frac{\mu\xi_q+\frac{\lambda\eta}{\prod_{i\neq K,\sigma(j),q}\xi_i} - \xi_q\xi_{\sigma(j)}}{\mu\xi_q+\frac{\lambda\eta}{\prod_{i\neq K,\sigma(j),q}\xi_i} - \frac{\eta}{\prod_{i \neq q,K}\xi_{i}} }\Bigg)  \nonumber \\
  &  & \Big(\prod_{i \neq k, \sigma^{-1}(q)} \xi_{\sigma(i)}^{x_i - x_k + y_K -y_{\sigma(i)} -1} \Big)~\eta^{x_k -y_K-1}\xi_q^{x_{\sigma^{-1}(q)} -x_k +y_K- y_q}~d\xi_1 \cdots d\xi_{K-1}d\eta. \nonumber
  \end{eqnarray}
  If an $S$-matrix in
   $$
   \prod_{\substack{i<j,k\\\sigma(i) > \sigma(j)}} S_{\sigma(i)\sigma(j)}\prod_{\substack{k<i<j\\\sigma(i) > \sigma(j)}} S_{\sigma(i)\sigma(j)}
   $$
   contains $\xi_q$, the $S$-matrix much be in the form of $S_{\alpha q}~ (\alpha> q)$ because $q = \min Q_{\sigma}$, and observe that the exponent of $\xi_q^{x_{\sigma^{-1}(q)} -x_k +y_K- y_q}$ is nonnegative. Hence, if we integrate with respect to $\xi_q$, there are singularities at
  \begin{equation*}
    \frac{  \frac{\eta}{\prod_{i \neq K,q}\xi_i}}{\mu+\frac{\lambda\eta}{\prod_{i\neq K,q}\xi_i}}
  \end{equation*}
  and
  \begin{equation*}
    \frac{\frac{\eta}{\prod_{i \neq K,q}\xi_i} - \frac{\lambda\eta}{\prod_{i\neq K,\sigma(j),q}\xi_i}}{\mu}
  \end{equation*}
  where  $j>k$, and possibly at
  \begin{equation*}
  \frac{\xi_{\alpha} - \mu}{\lambda\xi_{\alpha}}
  \end{equation*}
  from $S_{\alpha q}$ if $\lambda \neq 0$. If $\lambda =0$, there is no pole in $S_{\alpha q}$. By Lemma \ref{lemma1}, Lemma \ref{lemma2}, and Lemma \ref{lemma3} all these poles lie outside of $\mathcal{C}_{R_q}$, and hence the integral with respect to $\xi_q$ is zero. Therefore,
  \begin{equation*}
  \sum_{\substack{\sigma \in \mathbb{S}_K,\\ \sigma \neq id}} I(\sigma) = 0.
  \end{equation*}
  This completes the proof.
\end{proof}
\begin{remark}
If $p$ or $q$ is zero, the model is totally asymmetric in the sense that particles move only in one direction. However, unlike the TASEP, the transition probability is not expressed as a determinant as long as $\mu \neq 1$.
\end{remark}

\section{The position of the $m$th leftmost particle} \label{sectionMth}
Let $x_m(t)$ be the position of the $m$th leftmost particle at time $t$. In this section we find the probability that the $m$th leftmost particle is at $x$ at time $t$, which is denoted by $\mathbb{P}(x_m(t) = x)$.
\subsection{Notations}\label{assumption}
  Let $S = \{s_1,\cdots,s_k\} \subset \mathbb{N} $ with $s_i < s_{j}~(i<j)$ and  $\xi$ be a $k$-dimensional vector in $\mathbb{C}^k$. Define a $k$-variable function on $\mathbb{C}^k$
\begin{equation}
I(\xi;s_1,\cdots,s_k) = J(\xi;s_1,\cdots,s_k)\cdot \Big(\prod_{s \in S} \xi_{s} -1\Big) \label{1}
\end{equation}
 where
\begin{equation*}
J(\xi;s_1,\cdots,s_k) := \prod_{ i<j }\frac{\xi_{s_i} - \xi_{s_j}}{\mu + \lambda \xi_{s_i}\xi_{s_j} - \xi_{s_j}}\cdot \frac{1}{\prod_{s \in S} ( \xi_{s}-1)}
\end{equation*}
for $S$ with $|S| \geq 2$.  If $S$ is a singleton with $s \in S$, then we define
 $$J(\xi;s) = \frac{1}{\xi_s-1},$$
 and so $I(\xi;s)=1$. For the empty set, we define $I(\xi;\emptyset) =J(\xi;\emptyset) = 1$. \\ \\
 \noindent Let $U=\{u_1,\cdots, u_N\}\subset \mathbb{N}$ with $u_i < u_{j}~ (i<j)$ and suppose that  $f_U(\xi)= f(\xi_{u_1},\cdots,\xi_{u_N})$, as a function of $\xi_{u_i}$,  is analytic in the set $\{\xi_{u_i} \in \mathbb{C} : 0< |\xi_{u_i}| < R_{u_i}\} $
 for each $\xi_{u_i}$. For $S =\{s_1,\cdots, s_k\} \subset U$ let
  \begin{equation}
  f_U(\xi;s_1,\cdots,s_k) = f(\xi_{u_1},\cdots,\xi_{u_N})|_{\xi_{u_i} =1,~u_i \not\in S}.\label{f}
  \end{equation}
   Define a map $g_U : U \rightarrow \mathbb{Z}$ by $g_U(u_i) = i$. That is, $g_U$ maps the $i$th smallest element in $U$ to $i$.   Define
$$\Sigma_U[S] := \sum_{s \in S}g_U(s)$$
for $S \subset U$, and let $U_P = U \setminus P$ for $P \subset U \subset \mathbb{N}$.\\ \\
\noindent  Recall the notations in \cite{TW1} for the use of the identity (1.9) in \cite{TW1}.
\begin{equation*}
[N] = \frac{\mu^N - \lambda^N}{\mu - \lambda}
\end{equation*}
and
\begin{equation*}
[N]! = [N][N-1]\cdots [1], ~{N \brack m} = \frac{[N]!}{[m]![N-m]!}
\end{equation*}
with $[0]!=1.$

\subsection{The position of the rightmost particle at time $t$}\label{rightmost}
The way of finding $\mathbb{P}(x_N(t) = x)$  is the same as the way of finding $\mathbb{P}(x_1(t) = x)$ in the ASEP \cite{TW1}. That is, we sum the transition probability over all possible configurations.
Let $x_N = x$ and  $x_{N-i} = x-v_1 -\cdots - v_{i}$ for $ 1 \leq  i \leq N-1$. Then the sum of $P_Y(X;t)$ over all possible configurations is written as a multiple geometric series which is convergent
\begin{eqnarray*}
& & \sum_{v_1 = 0}^{\infty}\cdots\sum_{v_{N-1} = 0}^{\infty}\Big(\frac{1}{2\pi i} \Big)^N\int_{\mathcal{C}_{R_N}}\cdots \int_{\mathcal{C}_{R_1}} \\
& & \hspace{1cm}
\sum_{\sigma \in \mathbb{S}_N} A_{\sigma}\prod_{k=1}^{N-1}\Big(\prod_{i=1}^{N-k}\xi_{\sigma(i)}\Big)^{-v_k}\Big(\prod_i \xi_i^{x-y_i-1}e^{\varepsilon(\xi_i)t}\Big)~d\xi_1\cdots d\xi_N.
\end{eqnarray*}
Summing over each $v_i$, the integrand becomes
\begin{equation*}
\sum_{\sigma \in \mathbb{S}_N}A_{\sigma}\frac{\xi_{\sigma(1)}^{N-1}\xi_{\sigma(2)}^{N-2}\cdots \xi_{\sigma(N-1)}}{(\xi_{\sigma(1)}\cdots \xi_{\sigma(N)}-1)(\xi_{\sigma(1)}\cdots \xi_{\sigma(N-1)}-1)\cdots (\xi_{\sigma(1)}-1)} \Big(\prod_i\xi_i - 1 \Big)\prod_i\Big(\xi_i^{x-y_i-1}e^{\varepsilon(\xi_i)t}\Big),
\end{equation*}
and $A_{\sigma}$ is written as
\begin{equation*}
A_{\sigma} = \textrm{sgn}~\sigma \frac{\prod_{i<j}\big(\mu + \lambda\xi_{\sigma(i)}\xi_{\sigma(j)} - \xi_{\sigma(j)}\big)}{\prod_{i<j}(\mu + \lambda\xi_i\xi_j - \xi_j)}.
\end{equation*}
Letting $\xi_i \rightarrow \xi_{N-i+1}$ in the identity in (1.6) in \cite{TW1} and then applying the identity to
\begin{equation*}
\sum_{\sigma \in \mathbb{S}_N}\textrm{sgn}~\sigma \prod_{i<j}(\mu + \lambda\xi_{\sigma(i)}\xi_{\sigma(j)} - \xi_{\sigma(j)})
\frac{\xi_{\sigma(1)}^{N-1}\xi_{\sigma(2)}^{N-2}\cdots \xi_{\sigma(N-1)}}{(\xi_{\sigma(1)}\cdots \xi_{\sigma(N)}-1)(\xi_{\sigma(1)}\cdots \xi_{\sigma(N-1)}-1)\cdots (\xi_{\sigma(1)}-1)},
\end{equation*}
the sum becomes
\begin{equation*}
\mu^{N(N-1)/2} \frac{\prod_{i<j}(\xi_i - \xi_j)}{\prod_i (\xi_i-1)}.
\end{equation*}
With this result we obtain the formula for $\mathbb{P}(x_N(t) = x)$ of the MADM, which is to be compared with $\mathbb{P}(x_1(t) = x)$ of the ASEP \cite{TW1}.
\begin{proposition}\label{right}
 With $\mu,\mathcal{C}_{R_i}$ in Theorem \ref{theorem1} and $I(\xi;1,\cdots,N)$ given by (\ref{1}), we have
\begin{equation}
 \mathbb{P}(x_N(t) = x) = \frac{\mu^{N(N-1)/2}}{(2\pi i)^N}\int_{\mathcal{C}_{R_N}}\cdots \int_{\mathcal{C}_{R_1}} I(\xi;1,\cdots,N)\prod_i^N\Big(\xi_i^{x-y_i-1}e^{\varepsilon(\xi_i)t}\Big)~d\xi_1\cdots d\xi_N. \label{rightmosteq}
\end{equation}
\end{proposition}
\subsection{The position of the leftmost particle at time $t$}
\begin{lemma}\label{lemma6}
Let $S$ be a nonempty subset of $U=\{u_1,\cdots,u_N\} \subset \mathbb{N}$ with $u_i < u_{j}~(i<j)$. Then,
\begin{equation*}
\sum_{S} (-1)^{|S|+1}c^{\Sigma_U[S] - |S|} =1
\end{equation*}
where the sum runs over all $S$ and $c$ is a nonzero constant.
\end{lemma}
\begin{proof}
Let
$$\mathcal{S}_n^{(1)}:= \{S\subset U : |S| =n, ~u_1 \in S\}$$
and
$$\mathcal{S}_n^{(0)}:= \{S \subset U : |S| = n,~ u_1 \not\in S\}.$$
 Then, there is a bijection from $\mathcal{S}_n^{(1)}$ to $\mathcal{S}_{n-1}^{(0)}$ that maps $S \in \mathcal{S}_n^{(1)}$ to ${S}' =S\setminus \{u_1\}\in\mathcal{S}_{n-1}^{(0)}$. It is easy to check that
\begin{equation*}
(-1)^{|S|+1}c^{\Sigma_U[S] - |S|} + (-1)^{|S'|+1}c^{\Sigma_U[S']  - |S'|} =0
\end{equation*}
for $2 \leq n \leq N$ and the only remaining subset after summing over all nonempty subsets is $S=\{u_1\}$ for which
\begin{equation*}
 (-1)^{|S|+1}c^{\Sigma_U[S] - |S|} = 1.
\end{equation*}
\end{proof}
\begin{lemma}\label{lemma5}
Let $U = \{u_1,\cdots,u_N\} \subset \mathbb{N}$ with $u_i < u_j~(i<j)$.  Given a nonempty subset
 $P= \{p_1,\cdots,p_{K}\} \subsetneq U$ let $S= \{s_1,\cdots, s_k\} \subset U_P$ with $s_i < s_{j}~(i<j).$  Then, with $\mu$ and $\mathcal{C}_{R_i}$  in Theorem \ref{theorem1}, and $I(\xi;s_1,\cdots,s_k)$ and $f_U(\xi; s_1,\cdots,s_k)$ given by (\ref{1}) and (\ref{f}), respectively, we have
\begin{eqnarray}
& & \sum_{z=0}^{\infty} \int_{\mathcal{C}_{R_p}}\int_{\mathcal{C}_{R_{s_k}}}\cdots\int_{\mathcal{C}_{R_{s_1}}}
\Big(\prod_i^k\xi_{s_i}\Big)^z
I(\xi;s_1,\cdots,s_k)f_U(\xi; s_1,\cdots,s_k,p) ~d\xi_{s_1}\cdots d\xi_{s_k}d\xi_p \nonumber \\
&=& - \int_{\mathcal{C}_{R_p}}\int_{\mathcal{C}_{R_{s_k}}}\cdots\int_{\mathcal{C}_{R_{s_1}}}
J(\xi; s_1,\cdots,s_k)f_U(\xi; s_1,\cdots,s_k,p) ~d\xi_{s_1}\cdots d\xi_{s_k}d\xi_p  \nonumber \\
& & \hspace{2cm}+ ~(2\pi i)^{k}\Big(\frac{1}{\mu}\Big)^{k(k-1)/2} \int_{\mathcal{C}_{R_p}} f_U(\xi; p) ~d\xi_p. \label{sum1}
\end{eqnarray}
where $\int_{\mathcal{C}_{R_p}} = \int_{\mathcal{C}_{R_{p_K}}}\cdots\int_{\mathcal{C}_{R_{p_1}}}$ and $d\xi_p = d\xi_{p_1} \cdots d\xi_{p_K}$. Also,  $f_U(\xi;s_1,\cdots,s_k,p)$ and  $f_U(\xi;p)$ imply $f_U(\xi;s_1,\cdots,s_k,p_1,\cdots,p_K)$ and $f_U(\xi;p_1,\cdots,p_K),$ respectively.
\end{lemma}
\begin{proof} In the first step, we change the contour $\mathcal{C}_{R_{s_1}}$ to ~$\mathcal{C}_{r_{s_1}}$ with $r_{s_1} < ({1}/{R_{s_k}})^{k-1} <1$ in the left hand side of (\ref{sum1}). Then,
\begin{eqnarray*}
& &\int_{\mathcal{C}_{R_p}}\int_{\mathcal{C}_{R_{s_k}}}\cdots\int_{\mathcal{C}_{R_{s_1}}}
\Big(\prod_{i=1}^k\xi_{s_i}\Big)^z
I(\xi;s_1,\cdots,s_k)~f_U(\xi; s_1,\cdots,s_k,p) ~d\xi_{s_1}\cdots d\xi_{s_k}d\xi_p \\
&=&\int_{\mathcal{C}_{R_p}}\int_{\mathcal{C}_{R_{s_k}}}\cdots\int_{\mathcal{C}_{r_{s_1}}}
\Big(\prod_{i=1}^k\xi_{s_i}\Big)^z
I(\xi;s_1,\cdots,s_k)~f_U(\xi; s_1,\cdots,s_k,p) ~d\xi_{s_1}\cdots d\xi_{s_k}d\xi_p \\
& &  \hspace{0.2cm}+~ ~(2\pi i)\frac{1}{\mu^{k-1}}\int_{\mathcal{C}_{R_p}}\int_{\mathcal{C}_{R_{s_k}}}\cdots\int_{\mathcal{C}_{R_{s_2}}} \\
& & \hspace{3cm}
\Big(\prod_{i=2}^k\xi_{s_i}\Big)^z
I(\xi;s_2,\cdots,s_k)f_U(\xi; s_2,\cdots,s_k,p) ~d\xi_{s_2}\cdots d\xi_{s_k}d\xi_p
\end{eqnarray*}
where the second term on the right hand side is the residue at $\xi_{s_1} = 1$ and $\frac{1}{\mu^{k-1}}$ is from
$$\Bigg[\prod_{1<j}\frac{\xi_{s_1}-\xi_{s_j}}{\mu + \lambda \xi_{s_1}\xi_{s_j} - \xi_{s_j}}\Bigg]_{\xi_{s_1}=1}.$$
Note that if $\lambda \neq 0$, there are other poles of $\xi_{s_1}$ in $I(\xi;s_1,\cdots,s_k)$ but these poles lie outside of $\mathcal{C}_{R_{s_1}}$ by Lemma \ref{lemma1}. Hence these poles do not produce residues. If $\lambda =0$, there are no poles of $\xi_{s_1}$ in $I(\xi;s_1,\cdots,s_k)$.  In the second step, we change the contour $\mathcal{C}_{R_{s_2}}$ in the residue obtained from the first step to ~$\mathcal{C}_{r_{s_2}}$ with $r_{s_2} < ({1}/{R_{s_k}})^{k-2} <1$, which produces a new residue. We repeat this procedure until we obtain the integral with respect to only $\xi_{p_1}, \cdots, \xi_{p_K}$.
In general, in the $l$th step with $ l \geq 2$,  the contour $\mathcal{C}_{R_{s_l}}$ in the residue obtained in the $l-1$ step  is changed to ~$\mathcal{C}_{r_{s_l}}$ with $r_{s_l} < ({1}/{R_{s_k}})^{k-l} \leq 1$. Thus, after the $k$th step, the original integral is written as
\begin{eqnarray*}
\sum_{l=1}^k c_l\int_{\mathcal{C}_{R_p}}\int_{\mathcal{C}_{R_{s_k}}}\cdots\int_{\mathcal{C}_{r_{s_l}}}
\Big(\prod_{i=l}^k\xi_{s_i}\Big)^z
I(\xi;s_l\cdots,s_k)f_U(\xi; s_l\cdots,s_k,p) ~d\xi_{s_l}\cdots d\xi_{s_k}d\xi_p
\end{eqnarray*}
where
\begin{equation*}
c_l=
\begin{cases} 1  & \textrm{if ~$l=1$}, \\
  \frac{(2\pi i)^{l-1}}{\mu^{\sum_{i=1}^{l-1}(k-i)}}  & \textrm{if~ $2 \leq l \leq k$}.
\end{cases}
\end{equation*}
Now, the geometric series $\sum_{z=0}^{\infty}\Big(\prod_{i=l}^k\xi_{s_i}\Big)^z$ converges since $|\prod_{i=l}^k\xi_{s_i}|<1$. In the $l$th term where $1 \leq l \leq k$,
\begin{equation*}
\sum_{z=0}^{\infty}\Big(\prod_{i=l}^k\xi_{s_i}\Big)^z
I(\xi;s_l\cdots,s_k) =- J(\xi;s_l\cdots,s_k),
\end{equation*}
and then changing $\mathcal{C}_{r_{s_l}}$ back to $\mathcal{C}_{R_{s_l}}$, the $l$th term becomes
\begin{eqnarray*}
& & -c_l\int_{\mathcal{C}_{R_p}}\int_{\mathcal{C}_{R_{s_k}}}\cdots\int_{\mathcal{C}_{R_{s_l}}}
J(\xi;s_l\cdots,s_k)f_U(\xi; s_l\cdots,s_k,p) ~d\xi_{s_l}\cdots d\xi_{s_k}d\xi_p \\
& &+~ c_l
\frac{2\pi i}{\mu^{k-l}}\int_{\mathcal{C}_{R_p}}\int_{\mathcal{C}_{R_{s_k}}}\cdots\int_{\mathcal{C}_{R_{s_{l+1}}}}
J(\xi;s_{l+1}\cdots,s_k)f_U(\xi; s_{l+1}\cdots,s_k,p) ~d\xi_{s_{l+1}}\cdots d\xi_{s_k}d\xi_p
\end{eqnarray*}
and the $k$th term becomes
\begin{equation*}
 -c_k\int_{\mathcal{C}_{R_p}}\int_{\mathcal{C}_{R_{s_k}}}
J(\xi;s_k)f_U(\xi; s_k,p) ~d\xi_{s_k}d\xi_p + c_k(2\pi i)
\int_{\mathcal{C}_{R_p}}
f_U(\xi;p) ~d\xi_p.
\end{equation*}
 The positive term of the $l$th term cancels the negative term of the $(l+1)$th term for $l=1,\cdots,k-1$, and hence the only remaining terms are  the terms in the right hand side of (\ref{sum1}).
\end{proof}
\begin{proposition}\label{propo}
Let $1/2<\mu < 1$ and $U =\{u_1,\cdots,u_N\} \subset \mathbb{N}$ with $u_i < u_{j}~(i<j)$. Let  $A_{\sigma}$ be given by (\ref{coeff}) with a permutation  $\sigma$ on $ U$. With $\mathcal{C}_{R_{u_i}}$ as in Theorem \ref{theorem1} and  $f_U(\xi)$ given by (\ref{f}), we have for $2 \leq N  \in \mathbb{N}$,
\begin{eqnarray*}
& & \sum_{z_1,\cdots, z_{N-1} = 0}^{\infty}\Big(\frac{1}{2\pi i}\Big)^{N}\int_{\mathcal{C}_{R_{u_N}}}\cdots\int_{\mathcal{C}_{R_{u_1}}}\sum_{\sigma \in \mathbb{S}_N}A_{\sigma} \prod_{k=2}^N\Big(\prod_{i=k}^N \xi_{\sigma(i)} \Big)^{z_{k-1}} f_U(\xi)~ d\xi_{u_1}\cdots d\xi_{u_N} \\
&=&~
 \sum_{S }c_S\Big(\frac{1}{2\pi i}\Big)^{k}
\int_{\mathcal{C}_{R_{s_k}}}\cdots\int_{\mathcal{C}_{R_{s_1}}}I(\xi;s_1,\cdots,s_k)~f_U(\xi;s_1,\cdots,s_k)~
d\xi_{s_1}\cdots d\xi_{s_k}
\end{eqnarray*}
where
$$c_S = (-1)^{|S|+1}\frac{\lambda^{\Sigma_U[S] - |S|}}{\mu^{\Sigma_U[S] - |S|(|S|+1)/2}}
$$
and the sum runs over all  nonempty subsets $S = \{s_1,\cdots,s_k\} \subset U$ with $s_i < s_{j}~(i<j)$.
\end{proposition}
\begin{proof}
We prove by induction. Let $U=\{u_1,u_2\}$. Then,
\begin{eqnarray*}
 & &\Big(\frac{1}{2\pi i}\Big)^{2}\int_{\mathcal{C}_{R_{u_2}}}\int_{\mathcal{C}_{R_{u_1}}}\Big(\xi_{u_2}^z+ S_{u_2u_1} \xi_{u_1}^z\Big)f_U(\xi)~d\xi_{u_1}d\xi_{u_2} \\
 &=&\Big(\frac{1}{2\pi i}\Big)^{2}\int_{\mathcal{C}_{r_{u_2}}}\int_{\mathcal{C}_{R_{u_1}}}\xi_{u_2}^z f_U(\xi)~d\xi_{u_1}d\xi_{u_2}+
 \Big(\frac{1}{2\pi i}\Big)^{2}\int_{\mathcal{C}_{R_{u_2}}}\int_{\mathcal{C}_{r_{u_1}}}S_{u_2u_1} \xi_{u_1}^zf_U(\xi)~d\xi_{u_1}d\xi_{u_2}.
\end{eqnarray*}
Note that the first integral in the right hand side has been obtained from changing $R_{u_2}$ to  $r_{u_2}<1$ which is permissible  because the integrand is analytic in the circle with radius $R_{u_2}$ except at the origin, and the second integral in the right hand side has been obtained from changing $R_{u_1}$  to $r_{u_1} <1$ which is also permissible because  the pole from $S_{u_2u_1}$  lies outside of $\mathcal{C}_{R_{u_1}}$ by Lemma \ref{lemma1}. Summing over $z$ and then changing  contours back to the original contours, the first integral in the right hand side becomes
\begin{eqnarray}
 & &\Big(\frac{1}{2\pi i}\Big)^{2} \int_{\mathcal{C}_{r_{u_2}}}\int_{\mathcal{C}_{R_{u_1}}}\frac{1}{1-\xi_{u_2}} f_U(\xi)~d\xi_{u_1}d\xi_{u_2} \label{N22} \\
 &=& \Big(\frac{1}{2\pi i}\Big)^{2} \int_{\mathcal{C}_{R_{u_2}}}\int_{\mathcal{C}_{R_{u_1}}}\frac{1}{1-\xi_{u_2}} f_U(\xi)~d\xi_{u_1}d\xi_{u_2} + \Big(\frac{1}{2\pi i}\Big) \int_{\mathcal{C}_{R_{u_1}}} f_U(\xi;u_1)~d\xi_{u_1} \nonumber
\end{eqnarray}
and the second integral becomes
\begin{eqnarray}
& & \Big(\frac{1}{2\pi i}\Big)^{2}\int_{\mathcal{C}_{R_{u_2}}}\int_{\mathcal{C}_{r_{u_1}}}S_{u_2u_1} \frac{1}{1-\xi_{u_1}}f_U(\xi)~d\xi_{u_1}d\xi_{u_2} \label{N23} \\
&=&\Big(\frac{1}{2\pi i}\Big)^{2}\int_{\mathcal{C}_{R_{u_2}}}\int_{\mathcal{C}_{R_{u_1}}}S_{u_2u_1}\frac{1}{1-\xi_{u_1}} f_U(\xi)~d\xi_{u_1}d\xi_{u_2} +  \Big(\frac{1}{2\pi i}\Big)\Big(\frac{\lambda}{\mu}\Big)\int_{\mathcal{C}_{R_{u_2}}} f_U(\xi;u_2)~d\xi_{u_2}\nonumber
\end{eqnarray}
where $\lambda/\mu$ is due to evaluating $S_{u_2u_1}$ at $\xi_{u_1} =1$. Now the sum of double integrals  in the right hand sides of (\ref{N22}) and (\ref{N23}) is
\begin{equation*}
 -\lambda\Big(\frac{1}{2\pi i}\Big)^{2}\int_{\mathcal{C}_{R_{u_2}}}\int_{\mathcal{C}_{R_{u_1}}} \frac{(\xi_{u_1}-\xi_{u_2})(\xi_{u_1}\xi_{u_2}-1)}{(\mu + \lambda\xi_{u_1}\xi_{u_2} - \xi_{u_2})(\xi_{u_1}-1)(\xi_{u_2}-1)} f_U(\xi)~d\xi_{u_1}d\xi_{u_2}.
\end{equation*}
Hence  we have
\begin{eqnarray*}
& & -\lambda\Big(\frac{1}{2\pi i}\Big)^{2} \int_{\mathcal{C}_{R_{u_2}}}\int_{\mathcal{C}_{R_{u_1}}} I(\xi; u_1,u_2)f_U(\xi)~d\xi_{u_1}d\xi_{u_2}\\
&+& \Big(\frac{1}{2\pi i}\Big)\int_{\mathcal{C}_{R_{u_1}}} I(\xi;u_1)f_U(\xi;u_1)~d\xi_{u_1} +  \Big(\frac{\lambda}{\mu}\Big)\Big(\frac{1}{2\pi i}\Big)\int_{\mathcal{C}_{R_{u_2}}} I(\xi;u_2)f_U(\xi;u_2)~d\xi_{u_2}
\end{eqnarray*}
where $I(\xi;u_1) = I(\xi;u_2)=1$.\\
 \indent Suppose that the statement is true for $N=K-1$.
Pick up $p \in \{u_1,\cdots, u_N\}.$ We first sum over all $\sigma$ with $\sigma(1) = p$, that is, we compute
\begin{equation}
\sum_{z_1,\cdots, z_{N-1} = 0}^{\infty}\Big(\frac{1}{2\pi i}\Big)^{N}\int_{\mathcal{C}_{R_{u_N}}}\cdots\int_{\mathcal{C}_{R_{u_1}}}\sum_{\substack{\sigma ~\textrm{with} \\ \sigma(1) =p}}A_{\sigma} \prod_{k=2}^N\Big(\prod_{i=k}^N \xi_{\sigma(i)} \Big)^{z_{k-1}} f_U(\xi)~ d\xi_{u_1}\cdots d\xi_{u_N}. \label{eq10}
\end{equation}
Let us write $A_{\sigma}$ as
$$A_{\sigma} =\prod_{\substack{\sigma(j) < p,\\ j \geq 2}} S_{p\sigma(j)} \prod_{\substack{2\leq i<j,\\ \sigma(i) > \sigma(j)}} S_{\sigma(i)\sigma(j)}.$$
Changing the order of integration, which is permissible by Fubini theorem, we write (\ref{eq10}) as
\begin{eqnarray*}
& &\Big(\frac{1}{2\pi i}\Big)\int_{\mathcal{C}_{R_{p}}}\sum_{z_1=0}^{\infty}\Bigg[\sum_{z_2,\cdots, z_{N-1} = 0}^{\infty}\Big(\frac{1}{2\pi i}\Big)^{N-1}\int_{\mathcal{C}_{R_{\sigma(N)}}}\cdots\int_{\mathcal{C}_{R_{\sigma(2)}}}\sum_{\substack{\sigma ~\textrm{with} \\ \sigma(1) =p}}\prod_{\substack{2\leq i<j,\\ \sigma(i) > \sigma(j)}} S_{\sigma(i)\sigma(j)} \\
& & \hspace{1cm}\prod_{k=3}^N\Big(\prod_{i=k}^N \xi_{\sigma(i)} \Big)^{z_{k-1}}\prod_{\substack{\sigma(j) < p,\\ j \geq 2}} S_{p\sigma(j)}\Big(\prod_{u_i \neq p}\xi_{u_i}\Big)^{z_1}
f_U(\xi)~ d\xi_{\sigma(2)}\cdots d\xi_{\sigma(N)}\Bigg] d\xi_{p},
\end{eqnarray*}
where
\begin{equation*}
\prod_{\substack{\sigma(j) < p,\\ j \geq 2}} S_{p\sigma(j)} \Big(\prod_{u_i \neq p}\xi_{u_i}\Big)^{z_1}
f_U(\xi),
\end{equation*}
as a function of $\xi_{\sigma(2)}, \cdots, \xi_{\sigma(N)}$,
is analytic in the  set $\{\xi_{\sigma(i)} \in \mathbb{C} : 0< |\xi_{\sigma(i)}| < R_{\sigma(i)}\}$ for each $i=2,\cdots, N$ by Lemma \ref{lemma1} and the assumption of $f_U(\xi)$.
 Hence, by the induction hypothesis, we have
\begin{eqnarray}
& &\int_{\mathcal{C}_{R_{p}}}\sum_{z_1=0}^{\infty} \sum_{S \subset U_{P}}c_S\Big(\frac{1}{2\pi i}\Big)^{k+1}
\int_{\mathcal{C}_{R_{s_k}}}\cdots\int_{\mathcal{C}_{R_{s_1}}}I(\xi;s_1,\cdots,s_k) \prod_{\substack{\alpha < p,\\ \alpha \in S}} S_{p\alpha}
\nonumber \\
& &~\times \Big[\prod_{\substack{\alpha < p,\\ \alpha \not\in S}} S_{p\alpha} \Big]_{\xi_{\alpha}=1}\Big(\prod_{s \in S}\xi_{s}\Big)^{z_1}f_U(\xi)|_{\substack{\xi_i=1,\\ i \in U_{P} \setminus S}}~
d\xi_{s_1}\cdots d\xi_{s_k} d\xi_{p} \label{induction}
\end{eqnarray}
where $S = \{s_1,\cdots,s_k\}$ with $s_i < s_j ~(i<j)$ and  $P = \{p\}$.
Note that $c_S$ in (\ref{induction}) is defined for $S \subset U_{P}$.
If we denote the number of elements in the set $\{\alpha : \alpha < p ~~\textrm{and}~~ \alpha \in S\}$ by $l$, then
\begin{equation*}
  \Big[\prod_{\substack{\alpha < p,\\ \alpha \not\in S}} S_{p\alpha} \Big]_{\xi_{\alpha}=1} =
  \Big(\frac{\lambda}{\mu}\Big)^{g_U(p)-1-l}.
\end{equation*}
Hence (\ref{induction}) is written as
\begin{eqnarray}
& & \sum_{S \subset U_{P}}c_S\Big(\frac{\lambda}{\mu}\Big)^{g_U(p)-1-l}\Big(\frac{1}{2\pi i}\Big)^{k+1}\sum_{z_1=0}^{\infty}\int_{\mathcal{C}_{R_{p}}}
\int_{\mathcal{C}_{R_{s_k}}}\cdots\int_{\mathcal{C}_{R_{s_1}}}I(\xi;s_1,\cdots,s_k)
\prod_{\substack{\alpha < p,\\ \alpha \in S}} S_{p\alpha}\nonumber \\
& & \hspace{2cm}\times \Big(\prod_{s \in S}\xi_{s}\Big)^{z_1}f_U(\xi)|_{\substack{\xi_i=1,\\ i \in U_{P} \setminus S}}~
d\xi_{s_1}\cdots d\xi_{s_k} d\xi_{p}.\label{lemmaEq1}
\end{eqnarray}
Now, we can apply Lemma \ref{lemma5} to (\ref{lemmaEq1}) since $\prod_{{\alpha < p,~ \alpha \in S}} S_{p\alpha}\big[f_U(\xi)\big]_{\substack{\xi_i=1,\\ i \in U_{P}\setminus S}}$ satisfies the condition for the $f_U$ in Lemma \ref{lemma5}. Observing that  $ \Big[\prod_{\substack{\alpha < p,\\ \alpha \in S}} S_{p\alpha} \Big]_{\xi_{\alpha}=1}  = (\lambda/\mu)^l,$ we have, after applying Lemma \ref{lemma5},
\begin{eqnarray}
& &-  \sum_{S \subset U_{P}}c_S\Big(\frac{\lambda}{\mu}\Big)^{g_U(p)-1-l}\Big(\frac{1}{2\pi i}\Big)^{k+1}\int_{\mathcal{C}_{R_{p}}}
\int_{\mathcal{C}_{R_{s_k}}}\cdots\int_{\mathcal{C}_{R_{s_1}}}J(\xi;s_1,\cdots,s_k) \nonumber \\
& & \hspace{3cm}
\prod_{{\alpha < p,~ \alpha \in S}} S_{p\alpha}~  f_U(\xi)|_{\substack{\xi_i=1,\\ i \in U_{P} \setminus S}}~
d\xi_{s_1}\cdots d\xi_{s_k} d\xi_{p} \nonumber \\
& & +\sum_{S \subset U_{P}}c_S\Big(\frac{\lambda}{\mu}\Big)^{g_U(p)-1}\Big(\frac{1}{\mu}\Big)^{k(k-1)/2}\Big(\frac{1}{2\pi i}\Big)\int_{\mathcal{C}_{R_{p}}}
  f_U(\xi;p)~
 d\xi_{p} \label{eq15}
\end{eqnarray}
In the second sum,
$$
c_S\Big(\frac{\lambda}{\mu}\Big)^{g_U(p)-1}\Big(\frac{1}{\mu}\Big)^{k(k-1)/2} = (-1)^{|S|+1}\Big(\frac{\lambda}{\mu}\Big)^{\Sigma_{U_P}[S] - |S|+g_U(p)-1},
$$
and by Lemma \ref{lemma6}, the second sum becomes
\begin{equation}
 \Big(\frac{\lambda}{\mu}\Big)^{g_U(p)-1}\Big(\frac{1}{2\pi i}\Big)\int_{\mathcal{C}_{R_{p}}}
  f_U(\xi;p)~
 d\xi_{p} \label{p}.
\end{equation}
Now we  express $c_S\Big(\frac{\lambda}{\mu}\Big)^{g_U(p)-1-l} $ in the first sum
in terms of $\tilde{S} := P \cup S\subset U$.
\begin{eqnarray*}
 -c_S\Big(\frac{\lambda}{\mu}\Big)^{g_U(p)-1-l} &=& (-1)^{|S|}\frac{\lambda^{\Sigma_{U_{P}}[S] - |S|+g_U(p)-1-l}}{\mu^{\Sigma_{U_{P}}[S] - |S|(|S|+1)/2+g_U(p)-1-l}} \\
 &=&(-1)^{k}\frac{\lambda^{\sum_i^k g_{U_P}(s_i) -k+g_U(p)-1-l}}{\mu^{\sum_i^k g_{U_P}(s_i) - k(k+1)/2+g_U(p)-1-l}} \\
 &=&(-1)^{k}\frac{\lambda^{\sum_i^k g_{U_P}(s_i) +(k-l)+g_U(p) -2k-1}}{\mu^{\sum_i^k g_{U_P}(s_i) +(k-l)+g_U(p)- k- k(k+1)/2-1}} \\
 &=&(-1)^{|\tilde{S}|-1}\frac{\lambda^{\Sigma_U[\tilde{S}] - 2|\tilde{S}|+1}}{\mu^{\Sigma_U[\tilde{S}] - |\tilde{S}|(|\tilde{S}|+1)/2}} := \tilde{c}_{\tilde{S}}.
\end{eqnarray*}
Hence,   (\ref{eq15}) is written as a single sum
\begin{eqnarray}
& &\sum_{{S} \subset U_{P}}\tilde{c}_{\tilde{S}}\Big(\frac{1}{2\pi i}\Big)^{|\tilde{S}|}\int_{\mathcal{C}_{R_{p}}}
\int_{\mathcal{C}_{R_{s_k}}}\cdots\int_{\mathcal{C}_{R_{s_1}}}J(\xi;s_1,\cdots,s_k)  \nonumber \\
& & \hspace{2cm}
\prod_{ \alpha <p,~\alpha \in S} S_{p\alpha} ~ f_U(\xi)|_{\substack{\xi_i=1,\\ i \in U_{P} \setminus S}}~
d\xi_{s_1}\cdots d\xi_{s_k} d\xi_{p}  \label{eq11}
\end{eqnarray}
where $S$ is any subset of $U_P$. When $S $ is empty,  $\prod_{ \alpha <p,~\alpha \in S} S_{p\alpha}$ is defined to be 1  and the integral is over only $\xi_{p}$. Now, we sum (\ref{eq11}) over $p \in U$.
Observe  that for a nonempty $S \subset U_P$ with $|S| \geq 2$,
\begin{eqnarray}
J(\xi;s_1,\cdots,s_k)
\prod_{\substack{\alpha< p,\\ \alpha \in S}} S_{p\alpha} &=&\frac{\prod_{i<j}(\xi_{s_i} - \xi_{s_j})}{\prod_{i <j}(\mu + \lambda\xi_{s_i}\xi_{s_j} - \xi_{s_j})\prod_{s \in S}(\xi_{s} -1)}(-1)^l\prod_{\substack{ s<p,\\ s\in S}}\frac{(\mu + \lambda\xi_{s}\xi_p - \xi_{s})}{(\mu + \lambda\xi_{s}\xi_p - \xi_p)} \nonumber \\
& =& J(\xi;s_1,\cdots,s_k,p)(-1)^l\frac{\prod_{s \in S}(\mu + \lambda\xi_{s}\xi_p - \xi_{s})}{\prod_{\substack{s \in S,\\ s<p}}(\xi_{s} - \xi_{p})\prod_{\substack{s \in S,\\ s>p}}(\xi_{p} - \xi_{s})}~(\xi_{p} -1) \nonumber \\
&=&J(\xi;s_1,\cdots,s_k,p)\frac{\prod_{s \in S}(\mu + \lambda\xi_{s}\xi_p - \xi_{s})}{\prod_{s \in S}(\xi_{p} - \xi_{s})}~(\xi_{p} -1). \label{Mar1}
\end{eqnarray}
Note that if $|S|=1$, the last equality in (\ref{Mar1}) is still obtained.
Hence,  the sum of (\ref{eq11}) over $p \in U$ is written as
\begin{eqnarray*}
& &\sum_{{\tilde{S}} \subset U}\sum_{|S| =k-1}\tilde{c}_{\tilde{S}}\Big(\frac{1}{2\pi i}\Big)^{k}
\int_{\mathcal{C}_{R_{s_k}}}\cdots\int_{\mathcal{C}_{R_{s_1}}}J(\xi;s_1,\cdots,s_k)  \nonumber \\
& & \hspace{2cm} \prod_{s \in S,~ p \in \tilde{S}\setminus S}\Bigg[\frac{\mu + \lambda\xi_{s}\xi_p - \xi_{s}}{\xi_{p} - \xi_{s}}\Bigg]~(\xi_{p} -1) f_U(\xi)|_{\substack{\xi_i=1,\\ i \in U\setminus \tilde{S}}}~
d\xi_{s_1}\cdots d\xi_{s_k}  \label{eq34}
\end{eqnarray*}
where $\tilde{S}= \{s_1,\cdots,s_k\}$ with $ k \geq 2$ and $S \subset \tilde{S} \subset U$. Using the identity (1.9) in \cite{TW1} for the sum over all $S$ with $|S| = k-1$ in $\tilde{S}$, we finally obtain
\begin{equation}
\sum_{\tilde{S} \subset U}c_{\tilde{S}}\Big(\frac{1}{2\pi i}\Big)^{k}
\int_{\mathcal{C}_{R_{s_k}}}\cdots\int_{\mathcal{C}_{R_{s_1}}}I(\xi;s_1,\cdots,s_k) f_U(\xi;s_1,\cdots,s_k)~
d\xi_{s_1}\cdots d\xi_{s_k} \label{final}
\end{equation}
where
\begin{equation*}
 c_{\tilde{S}} := (-1)^{|\tilde{S}|+1}~\frac{\lambda^{\Sigma_U[\tilde{S}] - |\tilde{S}|}}{\mu^{\Sigma_U[\tilde{S}] - |\tilde{S}|(|\tilde{S}|+1)/2}}.
\end{equation*}
If $S$ is empty,   $\tilde{c}_{\tilde{S}} = {c}_{\tilde{S}}$  in (\ref{eq11}), and (\ref{eq11}) becomes
\begin{equation*}
c_{\tilde{S}}
\Big(\frac{1}{2\pi i}\Big)\int_{\mathcal{C}_{R_{p}}} f_U(\xi;p)~
d\xi_{p}
\end{equation*} for each $p$. This is combined with (\ref{final}) to  complete the proof.
\end{proof}
\begin{corollary}\label{corol}
Let $1/2 <\mu < 1$ and  $S=\{s_1,\cdots,s_k\}$ be a nonempty subset of $ \{1,\cdots, N\}$ with $s_i < s_{j}~(i<j)$ and $N \geq2$. In the MADM with $N$ particles the probability that  the leftmost particle is at $x$ at time $t$ is
\begin{equation}
\mathbb{P}(x_1(t) = x) =  \sum_{S }c_S\Big(\frac{1}{2\pi i}\Big)^k
\int_{\mathcal{C}_{R_{s_k}}}\cdots\int_{\mathcal{C}_{R_{s_1}}}I(\xi;s_1,\cdots,s_k)\prod_{s \in S}\Big(\xi_{s}^{x-y_{s}-1}e^{\varepsilon(\xi_{s})t}\Big)~
d\xi_{s_1}\cdots d\xi_{s_k} \label{leftmosteq}
\end{equation}
where
$$c_S = (-1)^{|S|+1}\frac{\lambda^{\Sigma[S] - |S|}}{\mu^{\Sigma[S] - |S|(|S|+1)/2}}.
$$
and $\Sigma[S]$ is the sum of all elements in $S$.
\end{corollary}
\begin{proof}
Choose $U = \{1,\cdots,N\}$ and $f_U(\xi) = \prod_i\xi_i^{x-y_i-1}e^{\varepsilon(\xi_i)t}$ in Proposition \ref{propo}. Then $\mathbb{P}(x_1(t) = x)$ is immediately obtained.
\end{proof}
\subsection{Main results : the position of the $m$th leftmost particle at time $t$}
\begin{theorem}\label{main}
 Let $1/2<\mu < 1$ and $U =\{u_1,\cdots,u_N\} \subset \mathbb{N}$ with $u_i < u_{j}~(i<j)$ for $N \geq 3$. Let $A_{\sigma}$ be given by (\ref{coeff}) with a permutation $\sigma$ on $U$. With $\mathcal{C}_{R_{u_i}}$ as in Theorem \ref{theorem1} and  $f_U(\xi)$ given by (\ref{f}), we have  for $1 < m <N$
\begin{eqnarray*}\label{propom}
& & \sum_{z_1,\cdots, z_{N-m} = 0}^{\infty}\sum_{v_1,\cdots, v_{m-1} = 0}^{\infty}\Big(\frac{1}{2\pi i}\Big)^N\int_{\mathcal{C}_{R_{u_N}}}\cdots\int_{\mathcal{C}_{R_{u_1}}}\sum_{\sigma \in \mathbb{S}_N }A_{\sigma} \prod_{k=1}^{m-1}\Big(\prod_{i=1}^{m-k} \xi_{\sigma(i)} \Big)^{-v_{k}}\\
& & \hspace{7cm} \prod_{k=m+1}^N\Big(\prod_{i=k}^N \xi_{\sigma(i)} \Big)^{z_{k-m}}f_U(\xi)~ d\xi_{u_1}\cdots d\xi_{u_N} \\
&=&~
 \sum_{|S| \geq m}c_S\Big(\frac{1}{2\pi i}\Big)^k
\int_{\mathcal{C}_{R_{s_k}}}\cdots\int_{\mathcal{C}_{R_{s_1}}}I(\xi;s_1,\cdots,s_k)f_U(\xi;s_1,\cdots,s_k)~
d\xi_{s_1}\cdots d\xi_{s_k}
\end{eqnarray*}
where
$$c_S = (-1)^{|S|+m}(\mu\lambda)^{m(m-1)/2}{|S| - 1 \brack |S|-m}\frac{\lambda^{\Sigma_U[S] - m|S|}}{\mu^{\Sigma_U[S] - |S|(|S|+1)/2}}
$$
and the sum runs over all $S = \{s_1,\cdots,s_k\} \subset U$ with $s_i < s_{j}~(i<j)$ and $k \geq m$.
\end{theorem}
\begin{proof}
Let $Q  \subset U$ such that  $|Q| = N-m$, and $P = U \setminus Q= \{p_1,\cdots, p_m\}$. Let $q $ be a permutation on $Q$, and let us  denote by $\mathcal{S}_q$  a set of $\sigma$ with $\sigma(m+i) = q(i)$ so that
 \begin{equation*}
Q=\{\sigma(m+1), \cdots, \sigma(N)\}
\end{equation*}
for $\sigma \in \mathcal{S}_q$.
 Observe that
 \begin{equation*}
  \sum_{\sigma \in \mathbb{S}_N} = \sum_{\{P,Q\}}\sum_{q \in \mathbb{S}_{N-m}}\sum_{\sigma \in \mathcal{S}_q}
 \end{equation*}
 where the first sum in the right hand side is over all partitions $\{P,Q\}$  of $U$ with $|P| = m$ and $|Q| = N-m$.
 We first sum over all permutations in  $\mathcal{S}_q$. Observe that geometric series over $v_i$ converge for all $i$ and $A_{\sigma}$ is written as
\begin{equation*}
A_{\sigma}  =\prod_{\substack{i<j \leq m, \\ \sigma(i)>\sigma(j)}}S_{\sigma(i)\sigma(j)}\prod_{\substack{i \leq m <j,\\ \sigma(i)> \sigma(j)}}S_{\sigma(i)\sigma(j)}\prod_{\substack{m <i<j,\\ \sigma(i) > \sigma(j)}}S_{\sigma(i)\sigma(j)}.
\end{equation*}
Since $\sigma(m+i)$ is fixed to $ q(i)$ for each $\sigma \in \mathcal{S}_q$, we have, after summing over all $v_i$,
\begin{eqnarray}
& &\sum_{z_1,\cdots, z_{N-m} = 0}^{\infty}\Big(\frac{1}{2\pi i}\Big)^N\int_{\mathcal{C}_{R_{u_N}}}\cdots\int_{\mathcal{C}_{R_{u_1}}} \label{equation}\\
& &\hspace{1cm} \Bigg[\sum_{\sigma \in S_q }
\prod_{\substack{i<j \leq m,\\ \sigma(i) > \sigma(j)}}S_{\sigma(i)\sigma(j)}
\frac{\xi_{\sigma(1)}^{m-1}\xi_{\sigma(2)}^{m-2}\cdots\xi_{\sigma(m-1)}}
{\big(\xi_{\sigma(1)}\cdots\xi_{\sigma(m-1)}-1\big)\big(\xi_{\sigma(1)}\cdots\xi_{\sigma(m-2)}-1\big)\cdots
\big(\xi_{\sigma(1)}-1\big)}\Bigg] \nonumber \\
& & \hspace{1cm}\prod_{\substack{i \leq m <j,\\ \sigma(i) > \sigma(j)}}S_{\sigma(i)\sigma(j)}\prod_{\substack{m <i<j,\\ \sigma(i) >\sigma(j)}}S_{\sigma(i)\sigma(j)}\prod_{k=m+1}^N\Big(\prod_{i=k}^N \xi_{\sigma(i)} \Big)^{z_{k-m}}f_U(\xi)~ d\xi_{u_1}\cdots d\xi_{u_N}, \nonumber
\end{eqnarray}
where  the sum over $\sigma \in \mathcal{S}_q$ is equal to
$$\mu^{m(m-1)/2}~I(\xi;\sigma(1),\cdots, \sigma(m))$$
 by the identity (1.6) in \cite{TW1}. Now, we sum over $q \in \mathbb{S}_{N-m}$. Let us write the sum of (\ref{equation}) over $q \in \mathbb{S}_{N-m}$ as
\begin{eqnarray*}
& &\mu^{m(m-1)/2}\sum_{z_1= 0}^{\infty}\sum_{z_2,\cdots, z_{N-m} = 0}^{\infty}\Big(\frac{1}{2\pi i}\Big)^N\int_{\mathcal{C}_{R_{u_N}}}\cdots \int_{\mathcal{C}_{R_{u_1}}}\\
& &\hspace{1cm}\sum_{q \in \mathbb{S}_{N-m}}\prod_{\substack{m <i<j,\\ \sigma(i) > \sigma(j)}}S_{\sigma(i)\sigma(j)} ~\prod_{k=m+2}^N\Big(\prod_{i=k}^N \xi_{\sigma(i)} \Big)^{z_{k-m}} \\
& &\hspace{1.5cm}I(\xi;p_1,\cdots, p_m)\prod_{i=1}^{N-m}\Big(\xi_{\sigma(m+i)}\Big)^{z_1}~f_U(\xi)\prod_{\substack{i \leq m <j,\\ \sigma(i) > \sigma(j)}}S_{\sigma(i)\sigma(j)}~ d\xi_{u_1}\cdots d\xi_{u_N}.
\end{eqnarray*}
Observe that
\begin{equation*}
 I(\xi;p_1,\cdots, p_m)\prod_{i=1}^{N-m}\Big(\xi_{\sigma(m+i)}\Big)^{z_1}~f_U(\xi)\prod_{\substack{i \leq m <j,\\ \sigma(i) > \sigma(j)}}S_{\sigma(i)\sigma(j)}
\end{equation*}
as a function of  $\xi_{\sigma(m+1)},\cdots, \xi_{\sigma(N)}$,
is analytic in the set $\{\xi_{\sigma(m+i)} \in \mathbb{C} : 0< |\xi_{\sigma(m+i)}| < R_{\sigma(m+i)}\} $
 for each $\xi_{\sigma(m+i)}~(i=1,\cdots,N-m)$ by the assumption on $f_U(\xi)$ and Lemma \ref{lemma1}. Hence, we may apply  Proposition \ref{propo} to the multiple integral over $\xi_{\sigma(m+1)},\cdots, \xi_{\sigma(N)}$. That is, if we change the order of integration which is permissible by Fubini's theorem,  we have
 \begin{eqnarray*}
& &\mu^{m(m-1)/2}\Big(\frac{1}{2\pi i}\Big)^N\int_{\mathcal{C}_{R_{p_m}}}\cdots \int_{\mathcal{C}_{R_{p_1}}}\sum_{z_1=0}^{\infty}\Bigg[\sum_{z_2,\cdots, z_{N-m} = 0}^{\infty}\int_{\mathcal{C}_{R_{\sigma(m+1)}}}\cdots \int_{\mathcal{C}_{R_{\sigma(N)}}} \\
& & ~
\sum_{q \in \mathbb{S}_{N-m}}\prod_{\substack{m <i<j,\\ \sigma(i) > \sigma(j)}}S_{\sigma(i)\sigma(j)}\prod_{k=m+2}^N\Big(\prod_{i=k}^N \xi_{\sigma(i)} \Big)^{z_{k-m}}
I(\xi;p_1,\cdots, p_m)\prod_{i=1}^{N-m}\Big(\xi_{\sigma(m+i)}\Big)^{z_1}~f_U(\xi)\\
& &\hspace{5.5cm}~\prod_{\substack{i \leq m <j,\\ \sigma(i) > \sigma(j)}}S_{\sigma(i)\sigma(j)}~ d\xi_{\sigma(N)}\cdots d\xi_{\sigma(m+1)}\Bigg]d\xi_{p_1}\cdots d\xi_{p_m},
\end{eqnarray*}
and applying Proposition \ref{propo}, we have
\begin{eqnarray}
&  &\mu^{m(m-1)/2}\Big(\frac{1}{2\pi i}\Big)^{m+k}\int_{\mathcal{C}_{R_{p_m}}}\cdots \int_{\mathcal{C}_{R_{p_1}}}\sum_{z_1=0}^{\infty} \sum_{S \subset Q}c_S\int_{\mathcal{C}_{R_{s_k}}}\cdots \int_{\mathcal{C}_{R_{s_1}}}\nonumber \\
& & \hspace{2cm}I(\xi;s_1,\cdots,s_k)\prod_{\substack{i \leq m <j,\\ \sigma(i)> \sigma(j) \in S}}S_{\sigma(i)\sigma(j)} ~\Big[\prod_{\substack{i \leq m <j,\\ \sigma(i)> \sigma(j) \in Q\setminus S}}S_{\sigma(i)\sigma(j)}\Big]_{\xi_{\sigma(j)}=1} \prod_{s \in S}\xi_{s}^{z_1}~\nonumber  \\
& & \hspace{2cm} I(\xi;p_1,\cdots, p_m)\big[f_U(\xi)\big]_{\substack{\xi_i = 1,\\ i \in Q \setminus S}}d\xi_{s_1}\cdots d\xi_{s_k}d\xi_{p_1}\cdots d\xi_{p_m} \label{eq28}
\end{eqnarray}
where $S = \{s_1,\cdots,s_k\}$ with $s_i < s_{j}~(i<j)$. Note that $c_S$ in (\ref{eq28}) is defined  for $S \subset Q$. For $p \in P$, let us denote the number of elements in the set $\{\alpha : \alpha \leq  p,~ \alpha \in P\}$ by $l_p^P$ and the number of elements in the set $\{\alpha: \alpha < p,~ \alpha \in S \subset Q\}$ by $l_p^S$. The number of elements in the set $\{\alpha: \alpha < p,~ \alpha \in Q\} = g_U(p) - l_p^P$ and hence
\begin{equation*}
\Big[\prod_{\substack{i \leq m <j,\\ \sigma(i)>\sigma(j) \in Q\setminus S}}S_{\sigma(i)\sigma(j)}\Big]_{\xi_{\sigma(j)}=1}
 = \Big(\frac{\lambda}{\mu}\Big)^{\sum_{p \in P}\big(g_U(p) - l_p^P - l_p^S\big)}.
\end{equation*}
 Hence, (\ref{eq28}) becomes
\begin{eqnarray*}
&  &\mu^{m(m-1)/2} \Big(\frac{1}{2\pi i}\Big)^{m+k}\sum_{S \subset Q}c_S \Big(\frac{\lambda}{\mu}\Big)^{\sum_{p \in P}\big(g_U(p) - l_p^P - l_p^S\big)}\sum_{z_1=0}^{\infty} \int_{\mathcal{C}_{R_{p_m}}}\cdots \int_{\mathcal{C}_{R_{p_1}}} \int_{\mathcal{C}_{R_{s_k}}}\cdots \int_{\mathcal{C}_{R_{s_1}}}\nonumber \\
& &I(\xi;s_1,\cdots,s_k) ~\prod_{\substack{i \leq m <j,\\ \sigma(i)> \sigma(j) \in S}}S_{\sigma(i)\sigma(j)} \prod_{s \in S}\xi_{s}^{z_1}~I(\xi;p_1,\cdots, p_m)\big[f_U(\xi)\big]_{\substack{\xi_i = 1,\\ i \in Q \setminus S}}d\xi_{s_1}\cdots d\xi_{s_k}d\xi_{p_1}\cdots d\xi_{p_m}.
\end{eqnarray*}
Now, we may apply Lemma \ref{lemma5} since
\begin{equation*}
\prod_{\substack{i \leq m <j,\\ \sigma(i)> \sigma(j) \in S}}S_{\sigma(i)\sigma(j)}
 ~I(\xi;p_1,\cdots, p_m)\big[f_U(\xi)\big]_{\substack{\xi_i = 1,\\ i \in Q \setminus S}}
\end{equation*}
 satisfies the condition for the $f_U$ in Lemma \ref{lemma5}. Hence, observing that
 \begin{equation*}
 \Big[\prod_{\substack{i \leq m <j,\\ \sigma(i) >\sigma(j) \in S}}S_{\sigma(i)\sigma(j)}\Big]_{\xi_{\sigma(j)}=1}  =\Big(\frac{\lambda}{\mu}\Big)^{\sum_{p \in P} l_p^S },
 \end{equation*}
  we have
\begin{eqnarray}
&- & \mu^{m(m-1)/2}\Big(\frac{1}{2\pi i}\Big)^{m+k}\sum_{S \subset Q}c_S\Big(\frac{\lambda}{\mu}\Big)^{\sum_{p \in P}\big(g_U(p) - l_p^P - l_p^S\big)}  \nonumber \\
& &
\int_{\mathcal{C}_{R_{p_m}}}\cdots \int_{\mathcal{C}_{R_{p_1}}}\int_{\mathcal{C}_{R_{s_k}}}\cdots\int_{\mathcal{C}_{R_{s_1}}} J(\xi;s_1,\cdots,s_k)\nonumber\\
& &
\hspace{1cm} \prod_{\substack{i \leq m <j,\\ \sigma(i) >\sigma(j) \in S}}S_{\sigma(i)\sigma(j)}
 ~I(\xi;p_1,\cdots, p_m)\big[f_U(\xi)\big]_{\substack{\xi_i = 1,\\ i \in Q \setminus S}}
d\xi_{s_1}\cdots d\xi_{s_k} d\xi_{p_1}\cdots d\xi_{p_m} \nonumber \\
&+ & \mu^{m(m-1)/2}\Big(\frac{1}{2\pi i}\Big)^{m}\sum_{S \subset Q}c_S\Big(\frac{\lambda}{\mu}\Big)^{\sum_{p \in P}\big(g_U(p) - l_p^P \big)} \Big(\frac{1}{\mu}\Big)^{k(k-1)/2}
  \nonumber \\
 & & ~~~~~\int_{\mathcal{C}_{R_{p_m}}}\cdots \int_{\mathcal{C}_{R_{p_1}}}I(\xi;p_1,\cdots, p_m)\big[f_U(\xi)\big]_{\substack{\xi_i = 1,\\ i \in Q }}~
d\xi_{p_1}\cdots d\xi_{p_m} \label{jan30}
\end{eqnarray}
by  Lemma \ref{lemma5}.
In the second sum of (\ref{jan30}),
$$
c_S\Big(\frac{\lambda}{\mu}\Big)^{\sum_{p \in P}\big(g_U(p) - l_p^P \big)} \Big(\frac{1}{\mu}\Big)^{k(k-1)/2} = (-1)^{|S|+1}\Big(\frac{\lambda}{\mu}\Big)^{\Sigma_Q[S] - |S|+\sum_{p \in P}\big(g_U(p) - l_p^P \big)},
$$
and  the second sum is written as
\begin{eqnarray*}
  & &\mu^{m(m-1)/2}\Big(\frac{1}{2\pi i}\Big)^{m} \Big(\frac{\lambda}{\mu}\Big)^{\sum_{p \in P}\big(g_U(p) - l_p^P \big)}\int_{\mathcal{C}_{R_{p_m}}}\cdots \int_{\mathcal{C}_{R_{p_1}}} \\
  & & \hspace{4cm}
  I(\xi;p_1,\cdots, p_m)\big[f_U(\xi)\big]_{\substack{\xi_i = 1,\\ i \in Q }}~
d\xi_{p_1}\cdots d\xi_{p_m}
\end{eqnarray*}
by Lemma \ref{lemma6}.
Now we  express $-c_S\Big(\frac{\lambda}{\mu}\Big)^{\sum_{p \in P}\big(g_U(p) - l_p^P - l_p^S\big)}  $ in the first sum of (\ref{jan30})
in terms of $\tilde{S} := P \cup S$. Observing that $\sum_{p \in P}l_p^P = m(m+1)/2$, we have
\begin{eqnarray*}
-c_S\Big(\frac{\lambda}{\mu}\Big)^{\sum_{p \in P}\big(g_U(p) - l_p^P - l_p^S\big)}  &=& (-1)^{|S|}\frac{\lambda^{\Sigma_Q[S] - |S|+\sum_{p \in P}\big(g_U(p) - l_p^P - l_p^S\big)}}{\mu^{\Sigma_Q[S] - |S|(|S|+1)/2+\sum_{p \in P}\big(g_U(p) - l_p^P - l_p^S\big)}} \\
 &=&(-1)^{|S|}\frac{\lambda^{\sum_i^k g_{Q}(s_i) -k+\sum_{p \in P}\big(g_U(p) - l_p^P - l_p^S\big)}}{\mu^{\sum_i^k g_{Q}(s_i) - k(k+1)/2+\sum_{p \in P}\big(g_P(p) - l_p^P - l_p^S\big)}} \\
 &=&(-1)^{|S|}\frac{\lambda^{\sum_i^k g_{Q}(s_i) +(mk-\sum_{p\in P}l_p^S) -(m+1)k+\sum_{p \in P}\big(g_U(p) - l_p^P \big)}}{\mu^{\sum_i^k g_{Q}(s_i) +(mk-\sum_{p\in P}l_p^S)- mk- k(k+1)/2+\sum_{p \in P}\big(g_U(p) - l_p^P \big)}} \\
 & =& (-1)^{|S|}\frac{\lambda^{\sum_i^k g_{Q}(s_i) +\sum_{p \in P} g_U(p)+ \sum_{p\in P}(k-l_p^S) -(m+1)(k+m)+m(m+1)/2}}{\mu^{\sum_i^k g_{Q}(s_i) +\sum_{p \in P} g_U(p)+\sum_{p\in P}(k-l_p^S)-(m+k)(m+k+1)/2}} \\
 &=&(-1)^{|\tilde{S}|-m}\frac{\lambda^{\Sigma_U[\tilde{S}] - (m+1)|\tilde{S}|+m(m+1)/2}}{\mu^{\Sigma_U[\tilde{S}] - |\tilde{S}|(|\tilde{S}|+1)/2}} := \tilde{c}_{\tilde{S}}.
\end{eqnarray*}
Hence,  (\ref{jan30}) is written as a single sum
\begin{eqnarray}
& &\mu^{m(m-1)/2}\Big(\frac{1}{2\pi i}\Big)^{|\tilde{S}|}\sum_{{S} \subset Q}\tilde{c}_{\tilde{S}}\int_{\mathcal{C}_{R_{p_m}}}\cdots\int_{\mathcal{C}_{R_{p_1}}}
\int_{\mathcal{C}_{R_{s_k}}}\cdots\int_{\mathcal{C}_{R_{s_1}}}J(\xi;s_1,\cdots,s_k)  \nonumber \\
& &
\prod_{\substack{i \leq m <j,\\ \sigma(i)>\sigma(j) \in S}}S_{\sigma(i)\sigma(j)} ~ I(\xi;p_1,\cdots, p_m)\big[f_U(\xi)\big]_{\substack{\xi_i = 1,\\ i \in Q \setminus S}}
d\xi_{s_1}\cdots d\xi_{s_k} d\xi_{p_1}\cdots d\xi_{p_m}  \label{jan301}
\end{eqnarray}
where  $S$  is any subset of $Q$. When $S$ is empty, $\prod_{\substack{i \leq m <j,\\ \sigma(i)> \sigma(j) \in S}}S_{\sigma(i)\sigma(j)}$ is defined to be 1 and  the integral is only over $\xi_{p_1}, \cdots, \xi_{p_m}$. Now, we sum (\ref{jan301}) over all partitions $\{P,Q\}$ of $U$.
Notice that if $S$ is nonempty, then
$$\prod_{\substack{i \leq m <j,\\ \sigma(i) > \sigma(j) \in S}}S_{\sigma(i)\sigma(j)} = \prod_{\substack{s < p,\\p\in P, s \in S}}S_{ps},$$
 and
\begin{eqnarray*}
& & J(\xi;s_1,\cdots,s_k)~I(\xi;p_1,\cdots, p_m)
\prod_{\substack{s < p,\\p\in P, s \in S}}S_{ps} \\
&=&\frac{\prod_{i<j}(\xi_{s_i} - \xi_{s_j})}{\prod_{i<j}(\mu + \lambda\xi_{s_i}\xi_{s_j} - \xi_{s_j})\prod_{s \in S}(\xi_{s} -1)}\frac{\prod_{p_i<p_j}(\xi_{p_i} - \xi_{p_j})\big(\prod_{p\in P}\xi_p -1\big)}{\prod_{p_i<p_j}(\mu + \lambda\xi_{p_i}\xi_{p_j} - \xi_{p_j})\prod_{p \in P}(\xi_{p} -1)}\\
& & \hspace{3cm}\times ~(-1)^{\sum_{p\in P}l_p^S}\prod_{\substack{s < p,\\p\in P, s \in S}}\frac{(\mu + \lambda\xi_{s}\xi_p - \xi_{s})}{(\mu + \lambda\xi_{s}\xi_p - \xi_p)}\\
& =& J(\xi;s_1,\cdots,s_k,p_1,\cdots,p_m)(-1)^{\sum_{p\in P}l_p^S} \Big(\prod_{p\in P}\xi_p -1\Big)\\
& & \hspace{2cm}\times \frac{\prod_{\substack{s < p,\\p\in P, s \in S}}(\mu + \lambda\xi_{s}\xi_p - \xi_{s})\prod_{\substack{s > p,\\p\in P, s \in S}}(\mu + \lambda\xi_{s}\xi_p - \xi_{s})}{\prod_{\substack{s < p,\\p\in P, s \in S}}(\xi_{s} - \xi_{p})\prod_{\substack{s > p,\\p\in P, s \in S}}(\xi_{p} - \xi_{s})} \\
&=&J(\xi;s_1,\cdots,s_k,p_1,\cdots,p_m)\frac{\prod_{s \in S, p\in P}(\mu + \lambda\xi_{s}\xi_p - \xi_{s})}{\prod_{s \in S, p \in P}(\xi_{p} - \xi_{s})}\Big(\prod_{p\in P}\xi_p -1\Big).
\end{eqnarray*}
Hence  the sum of (\ref{jan301}) over all partitions $\{P,Q\}$ of $U$ is written as
\begin{eqnarray}
& &\mu^{m(m-1)/2}\sum_{{\tilde{S}} \subset U}\sum_{|S|=k-m}\tilde{c}_{\tilde{S}}\Big(\frac{1}{2\pi i}\Big)^{k}
\int_{\mathcal{C}_{R_{s_k}}}\cdots\int_{\mathcal{C}_{R_{s_1}}}J(\xi;s_1,\cdots,s_k)  \nonumber \\
& & \hspace{1cm}\prod_{s \in S, p \in \tilde{S}\setminus S}\frac{(\mu + \lambda\xi_{s}\xi_p - \xi_{s})}{(\xi_{p} - \xi_{s})} \Big(\prod_{p\in \tilde{S}\setminus S}\xi_p -1\Big)f_U(\xi)|_{\substack{\xi_i=1,\\ i \in U\setminus \tilde{S}}}~
d\xi_{s_1}\cdots d\xi_{s_k} \nonumber \label{eq14}
\end{eqnarray}
where $\tilde{S}= \{s_1,\cdots,s_k\}$  with $k \geq m+1 $ and $S \subset \tilde{S} \subset U$. By using the identity (1.9) in \cite{TW1} for the sum over all $S$ with $|S| = k-m$, we obtain
\begin{equation}
\sum_{\tilde{S} \subset U}c_{\tilde{S}}\Big(\frac{1}{2\pi i}\Big)^{k}
\int_{\mathcal{C}_{R_{s_k}}}\cdots\int_{\mathcal{C}_{R_{s_1}}}I(\xi;s_1,\cdots,s_k) f_U(\xi;s_1,\cdots,s_k)~
d\xi_{s_1}\cdots d\xi_{s_k}\label{Final123}
\end{equation}
where
\begin{equation*}
 c_{\tilde{S}} = (-1)^{|\tilde{S}|+ m}(\mu\lambda)^{m(m-1)/2}~{k - 1 \brack k-m}\frac{\lambda^{\Sigma_U[\tilde{S}] - m|\tilde{S}|}}{\mu^{\Sigma_U[\tilde{S}] - |\tilde{S}|(|\tilde{S}|+1)/2}}.
\end{equation*}
If $S$ is empty, then   $\mu^{m(m-1)/2}\tilde{c}_{\tilde{S}}=c_{\tilde{S}}$ in (\ref{jan301}), and (\ref{jan301}) becomes
\begin{equation*}
c_{\tilde{S}}\Big(\frac{1}{2\pi i} \Big)^m \int_{\mathcal{C}_{R_{p_m}}}\cdots\int_{\mathcal{C}_{R_{p_1}}} I(\xi;p_1,\cdots,p_m)~f_U(\xi;p_1,\cdots,p_m)~d\xi_{p_1}\cdots d\xi_{p_m}
\end{equation*}
 for each $P$.
 This is combined with (\ref{Final123}) to complete the proof.
\end{proof}
\begin{theorem}\label{main2}
Let $1/2 < \mu < 1$ and  $S=\{s_1,\cdots,s_k\}$ be a nonempty subset of $ \{1,\cdots, N\}$ with $s_i < s_{j}~(i<j)$ and $N \geq2$. In the MADM with $N$ particles the probability that the $m$th leftmost particle  is at $x$ at time $t$ is
\begin{equation}
\mathbb{P}(x_m(t) =x ) =
 \sum_{|S| \geq m}c_S \Big(\frac{1}{2\pi i}\Big)^k
\int_{\mathcal{C}_{R_{s_k}}}\cdots\int_{\mathcal{C}_{R_{s_1}}}I(\xi;s_1,\cdots,s_k)\prod_{s\in S}\xi_{s}^{x-y_s-1}e^{\varepsilon(\xi_s)t}~
d\xi_{s_1}\cdots d\xi_{s_k} \label{maineq}
\end{equation}
where
$$c_S = (-1)^{|S|+m}(\mu\lambda)^{m(m-1)/2}{|S| - 1 \brack |S|-m}\frac{\lambda^{\Sigma[S] - m|S|}}{\mu^{\Sigma[S] - |S|(|S|+1)/2}}
$$
and $\Sigma[S]$ is the sum of all elements in $S$.
\end{theorem}
\begin{proof}
 When $m=1$ and $m=N$, (\ref{maineq}) is equal to (\ref{rightmosteq}) and (\ref{leftmosteq}), respectively.  When $ 3\leq N$, choose $U = \{1,\cdots,N\}$ and $f_U(\xi) = \prod_i\xi_i^{x-y_i-1}e^{\varepsilon(\xi_i)t}$ in Theorem \ref{main} for $1 <m <N$.
\end{proof}

\noindent Let  us denote the number of particles in $(-\infty,x]$ at time $t$ by $\mathcal{T}(x,t)$, which is called the \textit{total current} \cite{TW5}. The position of the $m$th leftmost particle at time $t$ is related to the total current. That is, from the fact that
\begin{equation*}
\{ \mathcal{T}(x,t) = m\} = \{x_m(t) \leq x, ~x_{m+1}(t)>x \}
\end{equation*}
it can be shown \cite{TW5} that
\begin{equation*}
\mathbb{P}( \mathcal{T}(x,t) \leq  m-1) = 1 -\mathbb{P}(x_m(t) \leq x).
\end{equation*}
Hence  summing (\ref{maineq}) over $x$ from $-\infty$ to $x$ to obtain $\mathbb{P}(x_m(t) \leq x)$, we have
\begin{eqnarray}
& & \mathbb{P}( \mathcal{T}(x,t) \leq  m-1)  \label{maineqcurrent} \\
&= &
 1-\sum_{|S| \geq m}c_S \Big(\frac{1}{2\pi i}\Big)^k
\int_{\mathcal{C}_{R_{s_k}}}\cdots\int_{\mathcal{C}_{R_{s_1}}}J(\xi;s_1,\cdots,s_k)\prod_{s\in S}\xi_{s}^{x-y_s}e^{\varepsilon(\xi_s)t}~
d\xi_{s_1}\cdots d\xi_{s_k}. \nonumber
\end{eqnarray}
\section{The two-sided PushASEP}\label{push}
The MADM is closely related to  the \textit{two-sided} PushASEP. In the \textit{one-sided} PushASEP the {pushing effect} is totally asymmetric and it is possible to express the transition probability as a determinant of contour integrals \cite{Borodin2} as in the TASEP. In the \textit{two-sided} PushASEP, unlike the \textit{one-sided} PushASEP,  the {pushing effect} is partially asymmetric, that is, the pushing effect is allowed in both directions. The dynamics of the \textit{two-sided} PushASEP with $N$ particles is as follows \cite{Ali2}: each particle is equipped with $2N$ Poisson clocks with rates $pr_n$ and $ql_n~ (p+q=1, 1\leq n\leq N),$ where $r_n$ and $l_n$ are given by (\ref{RightRate})  and (\ref{LeftRate1}), respectively. All clocks are independent and each site can be occupied by at most one particle. If the nearest empty site on the right (left) is $x + n~(x-n)$ when a clock with rate $pr_n ~(ql_n)$ of a particle at $x$ rings, then the particle at $x$ jumps to $x + n~(x-n)$. Otherwise, nothing happens and  the clock resumes. As mentioned in \cite{Borodin2} the jump to the right (left) of a particle to the nearest right (left) empty site is also interpreted as the particle's pushing all its right (left) neighboring particles by one if they prevent the particle from jumping to the right (left). The transition probability of the \textit{two-sided} PushASEP was treated in \cite{Lee2}, however, its proof is unfortunately false because the proof used Lemma 2.2 in \cite{TW1} which turned out to be false\footnote{The authors corrected the error in the erratum \cite{TW1}.}. In this section we  obtain the transition probability and the probability that the $m$th leftmost particle is at $x$ at time $t$ for the \textit{two-sided} PushASEP as  corollaries to Theorem \ref{theorem1} and Theorem \ref{main2}, respectively. In the PushASEP the physical region is
\begin{equation*}
\{(x_1,\cdots, x_N) \in \mathbb{Z}^N : x_1 < \cdots < x_N\}
\end{equation*}
because of the exclusion property. The $S$-matrix of the \textit{two-sided} PushASEP \cite{Ali2,Lee2} is defined to be
\begin{equation*}
S_{\beta\alpha}^{\dag} :=-\frac{\xi_{\beta}}{\xi_{\alpha}}\cdot\frac{\mu + \lambda \xi_{\alpha}{\xi_{\beta}} - \xi_{\alpha}}{\mu + \lambda \xi_{\alpha}{\xi_{\beta}} - \xi_{\beta}}=\frac{\xi_{\beta}}{\xi_{\alpha}} \cdot S_{\beta\alpha} \label{SmatrixofPush}
\end{equation*}
where $S_{\beta\alpha}$ is the $S$-matrix of the MADM and the Bethe ansatz solution is
\begin{equation*}
\sum_{\sigma \in \mathbb{S}_N}A_{\sigma}^{\dag}\prod_{i}^N\xi_{\sigma(i)}^{x_i}
\end{equation*}
with
\begin{equation*}
A_{\sigma}^{\dag} = \prod_{\substack{i<j,\\ \sigma(i) > \sigma(j)}} \frac{\xi_{\sigma(i)}}{\xi_{\sigma(j)}}\cdot S_{\sigma(i)\sigma(j)}.
\end{equation*}
Recalling Lemma 3 in \cite{Lee2}
\begin{equation*}
 \prod_{\substack{i<j,\\ \sigma(i) > \sigma(j)}} \frac{\xi_{\sigma(i)}}{\xi_{\sigma(j)}} = \prod_i \xi_{\sigma(i)}^{{\sigma(i)} -i},
\end{equation*}
we have the following transition probability of the \textit{two-sided} PushASEP.
\begin{corollary}
Let $\lambda + \mu=1,~ \frac{1}{2} < \mu  \leq 1$ and $\mathcal{C}_{R_i}~ (i=1,\cdots, N)$ be a circle oriented counterclockwise, centered at 0 with radius $R_i$. Assume that $1< R_{1} <\cdots < R_{N} < c$ where
\begin{equation*}
c = \begin{cases}
\frac{\mu}{\lambda} & \textrm{if}~~\lambda \neq 0\\
\infty &\textrm{if}~~ \lambda =0.
\end{cases}
\end{equation*}
 The transition probability of the \textit{two-sided} PushASEP with $N \geq 2$ particles is
\begin{equation}
P_Y(X;t) = \sum_{\sigma \in \mathbb{S}_N}\Big(\frac{1}{2\pi i}\Big)^N\int_{\mathcal{C}_{R_N}}\cdots\int_{\mathcal{C}_{R_1}}A_{\sigma}^{\dag}\prod_i^N \Big(\xi_{\sigma(i)}^{x_i - y_{\sigma(i)}-1}e^{\varepsilon(\xi_i)t}\Big)~d\xi_1\cdots \xi_N.\label{transitionPush}
\end{equation}
\end{corollary}
\begin{proof}
The integrand of (\ref{transitionPush}) is equal to
\begin{equation*}
A_{\sigma}\prod_i^N \Big(\xi_{\sigma(i)}^{x_i - y_{\sigma(i)}-1 + \sigma(i) -i}e^{\varepsilon(\xi_i)t}\Big)
\end{equation*}
by Lemma 3 in \cite{Lee2} where $A_{\sigma}$ is the  coefficient of the Bethe ansatz solution of the MADM. Since the only difference between (\ref{transition}) and (\ref{transitionPush}) is the exponent of $\xi_{\sigma(i)}$, we simply modify the proof for the MADM.
The exponent of $\xi_1$ in (\ref{Forpush}) should be replaced by $x_2 - x_1 + y_2-y_1 -2$ which is nonnegative in the physical region of the \textit{two-sided} PushASEP, and  the exponent of $\xi_q$ in (\ref{eqForPush}) should be replaced by
$$x_{\sigma^{-1}(q) }- x_k + y_K - y_q + q - \sigma^{-1}(q) - K+k$$
which is also nonnegative since $x_{\sigma^{-1}(q) }- x_k \geq {\sigma^{-1}(q) }- k $ and $y_K - y_q \geq K-q$ in the physical region of the \textit{two-sided} PushASEP. The rest of the proof is followed by the same argument as in the proof of the MADM.
\end{proof}
\noindent From the  alternate form of (\ref{transitionPush})
\begin{equation*}
P_Y(X;t) = \sum_{\sigma \in \mathbb{S}_N}\Big(\frac{1}{2\pi i}\Big)^N\int_{\mathcal{C}_{R_N}}\cdots\int_{\mathcal{C}_{R_1}}A_{\sigma}\prod_i^N \Big(\xi_{\sigma(i)}^{x_i -i-( y_{\sigma(i)} - \sigma(i))-1}e^{\varepsilon(\xi_i)t}\Big)~d\xi_1\cdots \xi_N,\label{transitionPush2}
\end{equation*}
it is confirmed that  the configuration $(x_i)_i$ in the \textit{two-sided} PushASEP is mapped to the configuration $(x_i-i)_i$ in the MADM. Hence, we immediately obtain $\mathbb{P}(x_m(t) = x)$ for the \textit{two-sided} PushASEP  by this mapping and (\ref{maineq}). That is, in the \textit{two-sided} PushASEP,
\begin{eqnarray}
& & \mathbb{P}(x_m(t) =x ) \label{maineqPush2}\\
& & \hspace{0.5cm}=~
 \sum_{|S| \geq m}c_S \Big(\frac{1}{2\pi i}\Big)^k
\int_{\mathcal{C}_{R_{s_k}}}\cdots\int_{\mathcal{C}_{R_{s_1}}}I(\xi;s_1,\cdots,s_k)\prod_{s\in S}\xi_{s}^{x-(y_s-s)-1}e^{\varepsilon(\xi_s)t}~
d\xi_{s_1}\cdots d\xi_{s_k}. \nonumber
\end{eqnarray}
\\
In the totally asymmetric limit $\lambda \rightarrow 0$,  only the subset $S= \{1,\cdots, m\}$ that  satisfies
\begin{equation*}
 \frac{m(m-1)}{2} + \Sigma[S] - m|S| =0
\end{equation*}
survives in (\ref{maineqPush2}). Hence, in this limit, we have
\begin{eqnarray*}
& & \mathbb{P}(x_m(t) =x ) \\
& & \hspace{0.5cm}=~
 \Big(\frac{1}{2\pi i}\Big)^k
\int_{\mathcal{C}_{R_{m}}}\cdots\int_{\mathcal{C}_{R_{1}}}I(\xi;1,\cdots,m)\prod_{i=1}^m\xi_{i}^{x-(y_i-i)-1}e^{\varepsilon(\xi_s)t}~
d\xi_{1}\cdots d\xi_{m}.
\end{eqnarray*}
This result tells us that the probability for the $m$th leftmost particle's position is determined by initial positions of the only first $m$ particles. This makes sense physically because in the  limit $\lambda \rightarrow 0$, particles are governed by the TASEP dynamics in the left direction but freely jump to the right neighboring site by pushing other particles if they occupy right sites.
\\ \\
\noindent\textbf{Acknowledgement} \\
 This work was supported by European Research Council.


\begin{thebibliography}{99}
\bibitem{Ali} {Alimohammadi, M., Karimipour, V. and Khorrami, M.:} \emph{Exact solution of a one-parameter family of asymmetric exclusion processes,} Phys. Rev. E, \textbf{57} 6370--6376 (1998).
\bibitem{Ali2} {Alimohammadi, M., Karimipour, V. and Khorrami, M.:} \emph{A two-parameteric family of asymmetric exclusion processes and its exact solution,} J. Stat. Phys., \textbf{97} 373--394 (1999).
\bibitem{Amir} {Amir, G., Corwin, I. and Quastel, J.:} \emph{Probability distribution of the free energy of the continuum directed random polymer in 1+1 dimensions,} Comm. Pure Appl. Math., \textbf{64} 466--537 (2011).
\bibitem{Borodin2} {Borodin, A. and Ferrari, P. L.:} \emph{Large time asymptotics of growth models on space-like paths $\Rmnum{1}$:PushASEP}, Electron. J. Probab., \textbf{13} 1380--1418 (2008).
\bibitem{Gwa} {Gwa, L. and  Spohn, H.:} \emph{Bethe solution for the dynamical-scaling exponent of the noisy Burgers equation}, Phys. Rev. A, \textbf{46} 844--854 (1992).
\bibitem{Lee} {Lee, E.:} \emph{Distribution of a particle's position in the ASEP with the alternating initial condition}, J. Stat. Phys., \textbf{140} 635--647 (2010).
\bibitem{Lee2} {Lee, E.:} \emph{Transition probabilites of the Bethe ansatz solvable interacting particle systems,} J. Stat. Phys., \textbf{142} 643--656 (2011).
\bibitem{Nagao} { Nagao, T. and Sasamoto, T.:} \emph{Asymmetric simple exclusion process and modified random matrix ensembles}, Nucl. Phys. B, \textbf{699} 487--502 (2004).
\bibitem{Povol} {Povolotsky, A. M.:} \emph{Bethe ansatz solution for the zero-range process with nonuniform stationary state},  Phys. Rev. E, \textbf{69} 061109 (2004).
\bibitem{Povol2} {Povolotsky, A. M.,  Priezzhev, V. B. and  Hu, C. -K.:} \emph{The asymmetric avalanche process}, J. Stat. Phys., \textbf{111} 1149--1182 (2003).
\bibitem{Rakos} { Rakos, A. and  Sch\"{u}tz, G. M.:} \emph{Current distribution and random matrix ensembles for an integrable asymmetric fragmmentation process,} J. Stat. Phys., \textbf{118} 511--530 (2005).
\bibitem{Sasamoto} {Sasamoto, T.:} \emph{Spatial correlations of the 1D KPZ surface on a flat substrate,} J. Phys. A, \textbf{38} L549--L556 (2005).
\bibitem{Sasamoto2} {Sasamoto, T. and  Wadati, M.:} \emph{Exact results for one-dimensional totally asymmetric diffusion models}, J. Phys. A, \textbf{31} 6057--6071 (1998).
\bibitem{SasaWada1998Oct} {Sasamoto, T. and Wadati, M.:} \emph{One-dimensional asymmetric diffusion model without exclusion}, Phys. Rev. E, \textbf{58} 4181--4190 (1998).
\bibitem{Schutz} { Sch\"{u}tz, G. M.:} \emph{Exact solution of the master equation for the asymmetric exclusion process}, J.  Stat. Phys., \textbf{88} 427--445 (1997).
\bibitem{TW1} {Tracy, C. A. and Widom, H.:} \emph{Integral formulas for the asymmetric simple exclusion process,} Commun. Math. Phys., \textbf{279} 815--844 (2008), {Erratum : Commum. Math. Phys. \textbf{304}, 875--878 (2011)}.
\bibitem{TW2} { Tracy, C. A. and  Widom, H.:} \emph{A Fredholm determinant representation in ASEP,} {J. Stat. Phys.}, \textbf{132} 291--300 (2008).
\bibitem{TW3} { Tracy, C. A. and  Widom, H.:} \emph{Asymptotics in ASEP with step initial condition,} Commun. Math. Phys., \textbf{290} 129--154 (2009).
\bibitem{TW4} { Tracy, C. A. and  Widom, H.:} \emph{On ASEP with step Bernoulli initial condition}, J. Stat. Phys., \textbf{137} 825--838 (2009).
\bibitem{TW5} { Tracy, C. A. and  Widom, H.:} \emph{Total current fluctuations in the asymmetric simple exclusion model}, J. Math. Phys., \textbf{50} 095204 (2009).


\end{thebibliography}
\end{document}